\documentclass[reqno,11pt]{amsart}

\usepackage{amssymb,graphics,graphicx,wrapfig}
\usepackage{color}

\def\eps{\varepsilon}

\def\N{\bbN}
\newcommand {\skalprod} [3] [] {\ensuremath{ \left\langle #2,#3 \right\rangle_{#1}}}

\def\be{\begin{equation}}
\def\ee{\end{equation}}
\def\ba{\begin{align}}
\def\bm{\begin{multline}}
\def\bfig{\begin{figure}[htb]}
\def\efig{\end{figure}}

\numberwithin{equation}{section}
\newtheorem{theorem}{Theorem}[section]
\newtheorem{proposition}[theorem]{Proposition}
\newtheorem{lemma}[theorem]{Lemma}

\newtheorem{definition}{Definition}

\DeclareMathSymbol{\leqslant}{\mathalpha}{AMSa}{"36}
\DeclareMathSymbol{\geqslant}{\mathalpha}{AMSa}{"3E}
\DeclareMathSymbol{\doteqdot}{\mathalpha}{AMSa}{"2B}
\DeclareMathSymbol{\circlearrowright}{\mathalpha}{AMSa}{"08}
\DeclareMathSymbol{\subsetneq}{\mathalpha}{AMSb}{"28}
\DeclareMathSymbol{\supsetneq}{\mathalpha}{AMSb}{"29}
\renewcommand{\leq}{\;\leqslant\;}
\renewcommand{\geq}{\;\geqslant\;}

\newcommand{\dd}{{\rm d}}
\newcommand{\e}[1]{\,{\rm e}^{#1}\,}
\newcommand{\ii}{{\rm i}}

\def\Tr{{\operatorname{Tr\,}}}

\def\Re{{\operatorname{Re\,}}}

\newcommand{\upchi}{\raise 2pt \hbox{$\chi$}}
\newcommand {\omitting}[1] {\backslash \hspace{-0.2 cm} #1}
\newcommand {\norm} [2] [] {\ensuremath{ \left\Vert  #2  \right\Vert_{#1} } } 

\def\writefig#1 #2 #3 {\rlap{\kern #1 truecm \raise #2 truecm
\hbox{#3}}}

\newcommand{\caC}{{\mathcal C}}

\newcommand{\caF}{{\mathcal F}}

\newcommand{\caH}{{\mathcal H}}

\newcommand{\caP}{{\mathcal P}}

\newcommand{\caS}{{\mathcal S}}

\newcommand{\bbC}{{\mathbb C}}

\newcommand{\bbN}{{\mathbb N}}

\newcommand{\bbR}{{\mathbb R}}

\newcommand{\bbZ}{{\mathbb Z}}

\newcommand{\bsj}{{\boldsymbol j}}

\newcommand{\bsI}{{\boldsymbol I}}

\newcommand{\bsP}{{\boldsymbol P}}

\newcommand{\bspsi}{{\boldsymbol\psi}}

\begin{document}


\title[Quantum oscillator in a photon field] {Effective density of states for
a quantum oscillator coupled to a photon field}

\author{Volker Betz and Domenico Castrigiano}
\address{Volker Betz \hfill\newline
\indent Department of Mathematics \hfill\newline
\indent University of Warwick \hfill\newline
\indent Coventry, CV4 7AL, England \hfill\newline
{\small\rm\indent http://www.maths.warwick.ac.uk/$\sim$betz/} 
}
\email{v.m.betz@warwick.ac.uk}

\address{Domenico Castrigiano \hfill\newline
\indent Fakult\"at f\"ur Mathematik \hfill\newline
\indent TU M\"unchen \hfill\newline
\indent Boltzmannstra\ss e 3, 85748 Garching, Germany \hfill\newline
{\small\rm\indent http://www-m7.ma.tum.de/bin/view/Analysis/DomenicoCastrigiano} 
}
\email{castrig@ma.tum.de}

\maketitle

\begin{quote}
{\small
{\bf Abstract.}
We give an explicit formula for the effective partition function of a 
harmonically bound particle minimally coupled to a photon field in the dipole approximation. 
The effective partition function is shown to be  
the Laplace transform of a positive Borel measure, the effective measure of states.
The absolutely continuous part of the latter allows for an analytic continuation, 
the singularities of which give rise to resonances. 
We give the precise location of these singularities, and show that they are well 
approximated by of first order poles with 
residues equal the multiplicities of the corresponding eigenspaces of the uncoupled quantum 
oscillator. Thus we obtain a complete analytic description of the natural line spectrum of the charged 
oscillator. 
}  

\vspace{1mm}

\vspace{1mm}
\noindent
{\footnotesize {\it 2010 Math.\ Subj.\ Class.:} 81V10, 82B10 }
\end{quote}


\section{Introduction}

The standard model for a non-relativistic quantum particle interacting with the radiation field is the 
Pauli-Fierz model using minimal coupling. The Hamiltonian is  
\be \label{PF}
H = \frac{1}{2 m} (p - e A)^2 + V + H_{\rm f},
\ee
where $H_{\rm f}$ is the Hamiltonian of the free field, $A$ is the transverse quantized field and $e$ the electron charge, $p$ is the 
particle momentum and $V$ the particle potential. We do not introduce form factors. 
Over the last years, there has been much activity and significant progress in understanding the Pauli-Fierz model. 
Notable developments include the proof of existence for a ground state when the infrared cutoff is removed 
\cite{GLL}, a detailed spectral analysis 
of the Hamiltonian and the proof of the existence of resonaces \cite{BFS98}, and a proof of enhanced binding through the 
photon field \cite{HS01}.
We refrain from giving a review of the by now vast literature on the subject, and instead refer to \cite{Sp04} 
and the bibliography therein. 

As is apparent already from the above selection of results, the majority of works on the subject investigates 
the spectral structure of the full system, particle and field. In our work, instead we regard the photon field as an 
environment in which the particle is embedded. A  prominent example of this point of view is Feynman's polaron \cite{F} 
by which our present investigations are inspired. 
The aim is  to treat the coupled particle as a closed system, and to derive, for the particle alone,  
effective equations  that capture the influence of the photon field. The particular effective quantity that we 
investigate (and, to our knowledge, indeed introduce) is the \textbf{effective measure of states}.

We consider a version of (\ref{PF}) that is simplified in two ways: Firstly, we use the dipole approximation, which amounts to replacing $A$ with $A(0)$, and secondly  we omit the  
self-interaction $A^2$ of the field. (The omission of the $A^2$ term is necessary only for a part of our calculations, and we 
discuss the possibility of keeping it in the final remark of Section \ref{model and results}.)
The particle potential $V$ is taken  to be harmonic, so that $H$ is completely quadratic, and we add the well-known terms of mass and energy renormalization. In this setup  we  compute  the ratio $Z$
 of the partition functions
of the full  system and of the free  field. We do this by first suitably discretizing both systems, 
and then taking the continuous-mode limit and the ultraviolet limit. This procedure yields a finite result. We find the explicit expression
\be \label{Z}
Z(\beta) = \left[\frac{\rho}{2\pi} \emph{e}^{-2 \rho \ln (\rho) \sin \varphi } \left\vert\Gamma(\ii \rho\,\emph{e}^{-\ii \varphi})\right\vert ^2   \right] ^3
\ee
with  $\rho \,\sim \beta$ 
inversely proportional to the temperature $T$
and sin$\varphi \sim e^2$ determining the strength of the coupling.  See Theorem \ref{Z(beta) UV singular} for details. 

From  the ratio $Z$  of the partition functions one gets an explicit expression for the difference of the respective free energies. Subtracting from this the free energy $F_0$ of the particle leads to the excess free  energy $F_{\rm ex} = -(\hbar /\beta)\ln(Z/Z_0)$. The latter   is particularly interesting in that  it is supposed to be experimentally accessible. Referring  to  \cite{CK87} we only mention that (\ref{Z}) confirms  the  well-known quadratic low-temperature behavior 
\be
F_{\rm ex}(\beta) \approx -\pi\alpha{\rm k}^2 T^2(3mc^2)^{-1}
\ee
 of the excess free energy.
 
In the present case, however, the significance of $Z$ goes beyond that.  As we will argue in Section \ref{subsys} for general coupled systems, 
\eqref{Z} is the effective partition function (or partition function for the particle subsystem). This means that $Z$  replaces $Z_0$ when the small system is regarded as autonomous while retaining an effective 
influence of the field. 
The central fact is that, as we prove, $Z$ is the 
Laplace transform of a positive Borel measure $\mu$. 
We interpret $\mu$ as the \textbf{effective measure of states} for the charged 
oscillator.
It has the same significance for the effective system that 
the measure of states $\mu_0$ (whose Laplace transform is $Z_0$, see \eqref
{Zbetaclassic} and \eqref{dsm discrete}),  
has for the uncoupled oscillator. It determines the probabilities of energy measurements at the system  in thermal 
equilibrium.

It is worth remembering  that in general the ratio 
of two partition functions is not the Laplace transform of some positive Borel measure. Most probably this is the case 
here  as long as  in the derivation of  \eqref{Z} the ultraviolet cutoff $C$ is  kept  finite (even if the 
continuous-mode limit is  already done).  But a finite cutoff in not desirable in any case, as
any  result we obtain for the effective system  would depend on the arbitrary parameter $C$. Thus it is crucial that the 
ultraviolet limit can be carried out, leading to an expression that \textbf{is} the Laplace transform of a positive 
measure.

Formula \eqref{Z}  allows a detailed study of the effective measure of states. 
First of all, it turns out that $\mu$ has one atom at 
the ground state frequency $\omega_\varphi$, is supported on the half line $[\omega_\varphi, \infty[$, and is otherwise absolutely continuous with respect to Lebesgue measure.  The corresponding 
density $\phi$ will be called \textbf{effective density of states}. We prove that $\phi$ has an analytic continuation  with singularities and cuts along the boundary  of a cone with apex $\omega_\varphi$ and 
opening angle  $2 \varphi$. The conjugate pairs of 
singularities  can be approximated by first order singularities in the interior of the cone, up to 
logarithmic corrections. The position and residue of these first order poles can be determined analytically. 
On the real line, they give rise to Lorentz profiles that replace the Dirac peaks present in $\mu_0$. 
The mass of each Lorentz profile is equal to the mass of the corresponding Dirac peak, 
and its position agrees with predictions of QED time-dependent perturbation theory 
up to first order in the fine structure constant $\alpha$. 
See Theorems \ref{ems structure 1} and \ref{ems structure 2} and the 
remarks following them for details.

From the above properties, a connection of the effective density of states with 
the theory of resonances appears obvious. The latter occur
when a small quantum system is coupled to a quantum field or reservoir, at which point eigenvalues 
of the small system dissolve into the continuous spectrum, and stationary states change into metastable states. 
In the Pauli-Fierz model, resonances have been investigated extensively by Bach, Fr\"ohlich and Sigal in
\cite{BFS98}. These authors employ the well-established complex dilation method.  
They study the spectrum of the full system, particle and field; resonances are then 
defined as the singularities of the analytic continuation of the resolvent to the second Riemann sheet.
In contrast, we work in the context of open quantum systems, and therefore the effective small system 
does not even have a Hamiltonian. While there has been work on resonances and decoherence in open quantum 
systems, such as \cite{BFS00} and \cite{MSB08-1,MSB08-2},
all authors seem to rely on the spectral structure of the full Hamiltonian.

We define  a resonance of an effective system in Definition \ref{effreso}, in terms of the complex structure of the 
effective density of states. For the harmonically bound charged particle, we show that this definition, and the concept 
of the effective density of states in general, produce physically reasonable results. Apart from  
facilitating explicit formulae, we do not expect the harmonic particle potential to be a vital 
part of the above picture; we thus conjecture that also for other confined charged systems, the effective density of 
states exists, and that its complex structure yields an appropriate description of the spectral lines by Lorentz 
profiles. A proof of this conjecture may well turn out to be challenging, and will probably require some insight into the 
connections of our results with the theory of \cite{BFS98}.

Our paper is organized as follows. In Section \ref{model and results}, we present our model in detail and 
state the main rigorous results of this work, along with a short discussion. Proofs are given in Sections
\ref{proofs} and \ref{proofs2}. The following Section \ref{concepts} interprets our results 
in the framework of quantum statistical mechanics, and discusses their connection 
to the theory of resonances.

\medskip
{\footnotesize
{\bf Acknowledgments:} We would like to thank Herbert Spohn for many useful discussions, and Marco Merkli for
helpful comments. 
V.B.\ is supported by the EPSRC fellowship EP/D07181X/1.
}


\section{Partition function of the embedded system}\label{concepts}
We regard the particle as embedded into the environment of the radiation  field. An environment is characterized by the fact that it stays in thermal equilibrium if the joint system is in 
thermal equilibrium. Roughly speaking, the interaction has effects on the embedded system but does not change 
the environment. 
More precisely,  when acting on the embedded system the environment carries out transitions between same states.

A well-known example of a quantum system whose embedding  in the environment  determines decisively 
its properties is the polaron, i.e.\ an electron in an ionic crystal. The computations in \cite[Chapter 8]{F} 
confirm that the electron moving with its accompanying distortion of the lattice behaves as a free particle but 
with an effective mass higher than that of the electron. In the same spirit 
we are going to treat the  problem of a charged oscillator surrounded by its own radiation field.  
In this section we put $\hbar=\rm k=1$.


\subsection{Energy distribution in thermal equilibrium}
Let the positive operator $H$ be the Hamiltonian of some quantum system.   For the moment we assume that 
 $\e{-\beta H}$ is trace class.
Classically, the partition function of the system  is given by
\be \label{Zbetaclassic}
Z(\beta) = \Tr \e{-\beta H}.
\ee
The inverse Laplace transform $\mu$ of $Z$, which we call  the \textbf{measure of states} of the quantum system, is the weighted sum of Dirac measures
\be \label{dsm discrete}
\mu = \sum_j N_j \delta_{\lambda_j},
\ee
where $\lambda_j$ are the eigenvalues of $H$ and $N_j$ are  the corresponding finite
multiplicities. Its physical relevance is due to the fact that the statistical uncertainty of 
states of the system in thermal equilibrium at the temperature $T=1/\beta$ is described by the probability measure $\mu^{(\beta)}:=\frac{1}{Z(\beta)}\e{-\beta(\cdot)}\mu$.
If $E$ denotes the spectral measure of $H$, and $\Delta$ a Borel subset of $\bbR$, then
\be \label{mainprop}
\mu^{(\beta)}(\Delta)=\frac{1}{Z(\beta)}\Tr(\e{-\beta H}E(\Delta))
\ee
is the probability for a measurement
of the energy to yield a value in $\Delta$. It is this interpretation that we 
will retain valid also for the effective entities.


\subsection{Spectral discretization}
From the basic  principle of statistical mechanics, once the partition function is known, 
all thermodynamic properties can be found.
In many interesting cases, however, including our  model Hamiltonian (\ref{basic H}), 
$\e{-\beta H}$ is not trace class. 
This raises the problem of how to attribute a partition function to 
the oscillator embedded in the photon field.
The way we solve this problem uses a slight generalization of methods 
from the theory of random Schr\"odinger operators (cf. \cite{KM07}).
The main idea is to approximate $H$ by operators $H_n$ with discrete spectrum. We define 

\begin{definition} \label{spec discret}
Let $H$ be a self-adjoint operator in a Hilbert space $\caH$, and let $(P_n)$ be a family of orthogonal projections on $\caH$. 
Put $H_n = P_n H P_n$. We say that $H_n$ is a {\bf spectral discretization} of $H$ (associated to $P_n$) if 
\begin{itemize}
\item[(i)] $P_n D(H) \subset D(H)$
\item[(ii)] $\lim_{n \to \infty} P_n = \rm 1$ in the strong topology.
\item[(iii)] The spectrum of $H_n$ acting in $P_n \caH$ consists of eigenvalues with finite multiplicity such that $\e{-\beta H_n}$ is trace class for all $n$. 
\end{itemize}  
\end{definition}


\subsection{Separation of a subsystem.}\label{subsys}
Let us now consider a general model consisting of two interacting quantum subsystems with respective state 
spaces  $\caH_{\rm 1}$ 
and  $\caH_{\rm 2}$. For convenience, the system corresponding to $\caH_1$ will be called the small system, the other one the large system; but note that there is actually no assumption on the dimension of 
$\caH_1$ and $\caH_2$. The Hamiltonian of the joint system, acting on $\caH_{\rm 1} \otimes \caH_{\rm 2}$ is 
\be \label{coup}
H = H_{\rm 1} \otimes 1 + 1 \otimes H_{\rm 2} + H_{\rm I},
\ee
where $H_{\rm I}$ describes the interaction. Again we assume for $H$ the existence of $\Tr \e{-\beta H}$, and 
also for $H_1$ and $H_2$.  Let the large system be  in thermal equilibrium.
Then the  influence of it onto the small system can be  described in a statistical way by averaging the states of the large system. Accordingly, the density operator attributed to the small system is $W(\beta):=(\Tr \e{-\beta H_2})^{-1}\Tr_2 \e{-\beta H}$, where $\Tr_2$ denotes the partial trace  with respect to the second factor. In the case $H_{\rm I}=0$ of non-interacting systems, $W(\beta)$ equals $\e{-\beta H_1}$. Moreover, if the joint system is in the  state $(\Tr \e{-\beta H})^{-1} \,\e{-\beta H}$, any prediction on measurements which concern only the small system is given by  the  thus uniquely determined state 
$(\Tr W(\beta))^{-1}W(\beta)=$
$(\Tr \e{-\beta H})^{-1}\Tr_2 \e{-\beta H}$
of the small system. This is due to the fact that 
$\Tr( \Tr_2 \e{-\beta H}\,E_1)=\Tr(\e{-\beta H}\;E_1 \otimes 1)$
holds for every orthogonal projection $E_1$ on $\caH_1$ and uniquely determines $\Tr_2$.

Consequently, we regard
$Z(\beta):=\Tr W(\beta)$
as the partition function attributed to the small system.  It satisfies
\be \label{Zbeta} 
Z(\beta)=\frac{\Tr \e{-\beta H}}{\Tr \e{-\beta H_2}},
\ee
and we will interpret its inverse Laplace transform as the effective measure of states (cf. \ref{dsm discrete}) 
for the small system when in contact with the large system. 
Of course, for this interpretation to make perfectly sense, 
$Z$ would need to be the Laplace transform of a positive Borel measure. This is not true in general, but 
we will find that it does hold in our model when performing the following limiting process.



\subsection{Infinitely large environment}
The large system in \eqref{PF} is described by $H_2 = H_{\rm f}$, which has absolutely continuous spectrum on 
$[0,\infty]$. Thus $\e{-\beta H_{2}}$ 
is not trace class, and neither  is $\e{-\beta H}$, whence  (\ref{Zbeta}) is not available. We consider 
spectral discretizations  
$H_n$ and $H_{2,n}$ of $H$ and $H_2$, and define 
\be \label{Znbeta}
Z_{n}(\beta) = \frac{\Tr \e{-\beta H_n}}{\Tr \e{-\beta H_{2,n}}}.
\ee
In the case we study 
$\lim_{n \to \infty} Z_n$
exists. After removing an ultraviolet cutoff we get the partition function $Z$, which is 
the Laplace transform of a positive measure $\mu$ with support in $[ 0,\infty [$. It is called the \textbf{effective measure of 
states} of the embedded system.  As in Section 2.1 we introduce the probability measures 
\be \label{pms}
\mu^{(\beta)}=\frac{1}{Z(\beta)}\e{-\beta(\cdot)}\mu
\ee
 for $\beta>0$, which are  fundamental in that they determine the probabilities of energy measurements:
if the charged oscillator is in thermal equilibrium  at the temperature $\frac{1}{\beta}$, then
$\mu^{(\beta)}(\Delta)$ is the probability for a measurement
of the energy to yield a value in $\Delta$.

One can read off $Z$ and $\mu$ the same amount of information about the spectrum of the 
embedded system  as in the case of a trace class operator. We will see that the embedding into its own 
radiation field  changes 
the behaviour of the oscillator qualitatively: instead of a purely discrete spectrum, we now obtain apart from 
the stable ground state an absolutely continuous 
spectrum on the positive  half axis.


\subsection{Resonances}\label{resonances}
An interesting property of the effective measure of states $\mu$ of  the present system is that the analytic continuation of its density function $\phi$
has conjugate pairs of first order singularities, which for small coupling are very close to the real line, cf. 
Theorems \ref{ems structure 1}, \ref{ems structure 2}. They manifest themselves as Lorentz profiles (Breit-Wigner resonance shapes) in the effective density of states. There is a connection with the theory of resonances \cite{H90}, 
which we will elucidate here. 

Generally speaking, bound states, which are perturbed, give rise to resonances. As shown in \cite{H90},  \cite{BFS98}, in case 
of a bounded electron coupled to a photon field
these are
related to singularities $p$ of the resolvent of the 
Hamiltonian $H$
off  the real axis coming from the eigenvalues of the uncoupled electron system $H_0$. 
In accordance with the definition of a \textbf{resonance}  given in \cite[XII.6]{RS}
 this means that there 
are  $p\in \bbC$ with $\operatorname{Im} \,p<0$ and a dense set  $\mathcal{D}$ of state vectors $u$ for which 
the matrix elements $R^{(u)}(z) = \langle u,(H-z)^{-1}u\rangle$ and $R^{(u)}_0(z)=\langle u,(H_0-z)^{-1}u\rangle$ have an 
analytic 
continuation from the upper complex half-plane across the positive real axis into the lower complex half-plane 
such that the points $p$ are singular points of the former but regular for the latter.

There is an equivalent formulation in terms of the scalar spectral measures. 
Let $E$ be the spectral measure of $H$. 
For given state vector $u$, denote by $\mu^{(u)}$ the \textbf{scalar spectral measure} on $[0,\infty[$ given by $\mu^{(u)}
(\Delta):=\langle u,E(\Delta)u\rangle$. Let $V$ be  the maximal open set on which the distribution function $w$ of $\mu^{(u)}$ is real analytic. We consider the complex analytic continuations $\phi^{(u)}$ of $w'\vert V$ on symmetric domains with respect to the real axis so that $\overline{\phi^{(u)}(z)}=\phi^{(u)}(\overline{z})$ holds. Let $\phi_0^{(u)}$ 
denote the respective object for $H_0$.

\begin{definition} \label{reso} 
Let $\mathcal{D}$ be a dense set of state vectors and $p=\hat{\omega} -\ii \frac{\gamma}{2} \in\bbC$ with $
\hat{\omega}>0$, $\gamma>0$. Then $p$ is a \textbf{resonance}   if for every $u\in\mathcal{D}$ there exist 
complex analytic continuations $\phi^{(u)}$ and $\phi_0^{(u)}$ such that $p$ is a singular point of the former and a 
regular point of the latter. 
\end{definition}

In the Appendix it is shown that  the two definitions of a resonance  are equivalent. More precisely,  
$R^{(u)}$ is analytic on $\bbC\setminus [0,\infty[$ and its  analytic continuation across the positive real axis to the second Riemann sheet equals $R^{(u)}+2\pi \ii \phi^{(u)}$. The same holds true for
the respective objects for $H_0$.

The singularities of $\phi^{(u)}$
occur in complex conjugate pairs. They manifest themselves as Lorentz profiles in $\phi^{(u)}$ along the real axis. For an illustration  assume the simplest case that the resonance  $p$ is a pole  of first order and that there are no other 
singularities than $p$ and its conjugate $\overline{p}$
in some open rectangle $U$, which is symmetric with respect to the real axis and the axis through $p,\,
\overline{p}$. Then there is a holomorphic function $\varphi^{(u)}$ on $U$ such that $\phi^{(u)}(z)=
\frac{\varphi^{(u)}(z)}{(z-p)(z-\overline{p})}$ for $z\in U\setminus\{p,\overline{p}\}$. Hence along the real axis not far from  $\hat{\omega}$  one has
\[
\phi^{(u)}(\omega)\simeq {\rm constant} \frac{\gamma /2}{(\omega-\hat{\omega})^2+\gamma^2 /4},
\]
which is a \textbf{Lorentz profile}. A quasi energy eigenstate $u$, which in the energy representation shows a narrow 
frequency band $\Delta$ around $\hat{\omega}$, decays exponentially. Indeed,
the matrix element $\langle u,\e{-\ii t H}u\rangle$ is given by
 the Fourier transform $\hat{\mu}^{(u)}
(t)=\int\e{-\ii t \omega}\dd\mu^{(u)}(\omega)$ 
of the scalar spectral measure, which in the present case 
becomes

\[
\hat{\mu}^{(u)}(t)=\int_\Delta \e{-\ii t\omega}\phi(\omega)\dd \omega\simeq \, {\rm constant} \, \e{-\ii \hat{\omega}t}\e{-t
\gamma/2}.
\]

The matrix  element $R^{(u)}(z)$
of the resolvent equals the  Stieltjes transform of the scalar spectral measure
\[
\tilde{\mu}^{(u)}(z):=\int \frac{\dd\mu^{(u)}(\omega)}{\omega-z}.
\]
The formula
\[
\hat{\mu}(t)=-\operatorname{lim}_{\epsilon\to 0+}\frac{1}{\pi}\int_{-\infty}^\infty\e{-\ii t\omega}\operatorname{Im}\tilde{\mu}(\omega+\ii \epsilon)\dd \omega,
\]
which converts the Stieltjes transform of any finite Borel measure $\mu$ on $[0,\infty[$ into its Fourier  transform,
allows to study the decay of a state $u$ related to a resonance starting from $R^{(u)}$.
Cf. \cite[I.32]{BFS98} and \cite{H90}.

\bigskip
To summarize so far, we have seen that it is quite natural (and equivalent) to determine the resonances by the analytic continuation of the densities of  scalar spectral measures of the Hamiltonian rather than 
 by the analytic continuation of  matrix elements of its resolvent. Now the question is how to apply these considerations to an effective system for which there is no Hamiltonian. Note that once the effective measure of states $\mu$ is determined, the probability distributions $\mu^{(\beta)}$  of the energy are at our disposal for every temperature $1/\beta$ (see Section 2.4). These are the primary  entities in case of an effective system. It appears natural that they have to replace the scalar spectral measures of the Hamiltonian. Hence common singular points (in the lower complex half-plane) of the analytic continuations $\phi^{(\beta)}$ of the respective densities are attributed to resonances.  It is immediate from  (\ref{pms})  that the former are exactly the singular points of the analytic continuation $\phi$ of the density regarding the effective measure of states $\mu$. Thus we are led to  the following
\begin{definition} \label{effreso} 
Resonances of an effective system are the singularities (in the lower complex half-plane) of the analytic continuation of the effective  density of states.
\end{definition} 
   We have seen  that resonances according  Definition \ref{effreso}   give rise to Lorentz profiles along the real axis representing the natural lines of the system. The results in Theorems \ref{ems structure 1} and \ref{ems structure 2}  then confirm  that these resonances actually occur for the charged harmonic oscillator.  What is more, they furnish a  physically  verifiable picture of the spectrum  in its entirety.


\section{Model and main results} \label{model and results}

\subsection{Particle coupled to a photon field}
Our starting point is the minimal coupling without a form factor  
(cf. \cite[II.D.1\,(D.1)]{CDG87}), in the dipole 
approximation 
of a charged particle interacting with a photon field and  in a  quadratic external potential. 
The Hamiltonian is given by 
\be \label{basic H}
H = (H_{\rm p}  + H_{\rm r}) \otimes {\rm 1} + {\rm 1}  \otimes H_{\rm f} + H_{\rm I} 
\ee
acting in $L^2(\bbR^3) \otimes \caF^{\otimes 2}$. Let us explain the symbols above. 
\[
H_{\rm p} = -\hbar^2 \frac{m}{2} \eta^2 \Delta_q + \frac{1}{2m} q^2,
\]
acting in $L^2(\bbR^3)$, is the energy of the harmonically bound particle, $m$ is the particle mass, and $\eta$ is the 
frequency associated to the the harmonic external potential. Note that we represent the 
particle in momentum space, which will later turn out to be convenient.
\[
H_{\rm f} =  H_{{\rm f},1} \otimes \rm1+ \rm1 \otimes H_{{\rm f},2}\;\; \mbox{ and }\;\; H_{{\rm f},\sigma} =\hbar  \int  \omega(k) a^{\ast}(k,\sigma) a(k,\sigma) \, \dd k, \quad \sigma=1,2
\]
is the energy of the free field. 
Each $H_{{\rm f},\sigma}$, $\sigma=1,2$, acts in the symmetric Fock  space 
$\caF = \bigoplus_{n=0}^{\infty} \caF^{(n)}$, where $\caF^{(n)}$ is the set of all   
$f \in L^2(\bbR^{3n})$ such that $f(k_1, \ldots, k_n) = 
f(k_{\pi(1)}, \ldots, k_{\pi(n)})$ for all permutations $\pi$, and all $k_j \in \bbR^3$.
We write an element $F$ of the Fock space as $F=\sum_{n=0}^\infty F^{(n)}$ with  $\sum_{n=0}^\infty \Vert F^{(n)}\Vert^2<\infty$ for $F^{(n)} \in \caF^{(n)}$.
We will need the following well-known fact: 
if $\{ f_m : m \in \bbN \}$ is a orthonormal basis of $L^2(\bbR^3)$ or a subspace $X$ thereof, then the system 
\[
\{ f_{m_1} \hat \otimes \ldots \hat\otimes f_{m_n} : m_1, \ldots, m_n \in \N, n \in \N \}
\]
is a orthonormal basis of the Fock space over $L^2(\bbR^3)$ or $X$, respectively. Above, $\hat\otimes$ denotes the 
symmetric tensor product.

By definition, $H_{\rm f, \sigma}$ acts on $\caF^{(n)}$ as 
\[
 \left(\hbar  \int \omega(k) a^\ast(k,\sigma) a(k,\sigma) \dd k \, F^{(n)} \right) (k_1, \ldots, k_n) = \sum_{j=1}^n\hbar \omega(k_j) F^{(n)}(k_1, \ldots, k_n),
\]
where $\omega(k) = c |k|$  is the photon dispersion relation with $c$ the velocity of light. The interaction between the particle and the field is given by 
\[
H_{\rm I} =  H_{{\rm I},1} \otimes {\rm 1} + {\rm 1} \otimes H_{{\rm I},2}\;\; \mbox{ and }\;\; H_{{\rm I},\sigma} 
= \frac{e}{m}  \int \chi_c(k) \sqrt{\frac{\hbar}{4\pi^2 \omega(k)}} (q \cdot u_\sigma(k)) 
\Big( a^\ast(k,\sigma) + a(k,\sigma) \Big) \dd k,
\]
with an ultraviolet cutoff implemented through 
$\chi_{c} = 1_{\{|k| <C\}}$. It is well known (see e.g.\ \cite{BHLMS02}) that $H_{\rm I}$ is infinitesimally form bounded with respect to $H_{\rm p} \otimes {\rm 1} + {\rm 1} \otimes H_{\rm f}$, so that $H$ is essentially self-adjoint on $D(H_{\rm p} \otimes {\rm 1} + {\rm 1} \otimes H_{\rm f})$. 

The  electron charge $e$ determines the strength of the coupling, and the transverse vector fields $u_\sigma$, 
such that $u_1(k)$, $u_2(k)$, $k/|k|$ are orthonormal, account for photon polarization.
$H_{\rm I}$ acts on $L^2(\dd q)$ as a 
multiplication operator, and on $\caF$ the creation and annihilation operators $a^\ast(g) \equiv \int g(k) a^\ast(k) \, \dd k$ and $a(g) \equiv \int g(k) a(k) \, \dd k$  act in the following way:
The creation operator takes $\caF^{(n)}$ into $\caF^{(n+1)}$ through 
\[
(a^\ast(g) F)^{(n+1)}(k_1, \ldots, k_{n+1}) = \frac{1}{\sqrt{n+1}} \sum_{j=1}^{n+1} g(k_j) F^{(n)}(k_1, \ldots, 
\omitting{k_i}, \ldots, k_{n+1}),
\]
where $\omitting{k_j}$ means that the argument $k_j$ is omitted. The annihilation operator takes $\caF^{(n)}$ into $\caF^{(n-1)}$ 
through 
\[
(a(g) F)^{(n-1)}(k_1, \ldots, k_{n-1}) = \sqrt{n} \int g(k) F^{(n)}(k,k_1,\ldots,k_{n-1}) \dd k.
\]
Finally,
\[
\begin{split}
H_{\rm r}& = \frac{e^2 \eta^2}{2 \pi^2 m} \sum_{\sigma=1,2} \int  \frac{\chi_c(k)}{\omega^2(k)} 
(u_\sigma(k)\cdot q)^2 \, \dd^3 k\, -\, \hbar \frac{e^2 \eta^2}{4\pi^2 m} \int \frac{\chi_c(k)}{\omega^2(k) (\omega(k) + \eta)} \, \dd^3 k\\
        & =\frac{4}{3\pi c^2}(e^2\eta^2/m)\,C\, q^2\,-\,\frac{\hbar}{\pi c^3}(e^2 \eta^2/m)\,\text{ln}(1+\frac{cC}{\eta})
\end{split}
\]

is a renormalization acting nontrivially only on $L^2(\dd q)$. The first term yields the standard mass renormalization 
accounting for the increased mass of the particle due to the coupling to the photon field, and
the second term    
is the well-known energy renormalization, which
compensates the negative shift   of the energy scale also due to the interaction with the photons.  Including these two terms will keep finite results when we remove the ultraviolet cutoff by taking 
$C\to \infty$.

We now give spectral discretizations for $H$ and $H_{\rm f}$. 
For $N \in \bbN$ let $ \Lambda_j,\, j \in \bbN$, be the  half open cubes with volume $V = N^{-3}$ and vertices in the lattice $(\frac{1}{N}\bbZ)^3\subset \bbR^3$. We define the one dimensional orthogonal projections 
\be \label{h_j}
P_N(j): L^2(\bbR^3) \to L^2(\bbR^3), \quad f \mapsto  \skalprod{h_j}{f} h_j \, \, \text{ with } h_j := \frac{1}{\sqrt{V}} 1_{\Lambda_j}.
\ee
Let  $J_N := \{j \in \N: \Lambda_j^{(N)} \cap B(0,N) \neq \emptyset \}$, where $B(0,N)$ is the centered ball with radius $N$. Since $P_N(j)\,P_N(j')=0$ for $j \not=j'$,
\[
P_N := \sum_{j \in J_N} P_N(j)
\]
is a finite dimensional orthogonal projection. Hence so is $\bsP_N$,
the second quantisation of $P_N$, acting on 
$\caF^{(n)}$ as 
\be \label{P_N}
\bsP_N\, f_1\hat \otimes \dots \hat \otimes f_n= (P_N f_1) \hat \otimes \dots \hat \otimes (P_Nf_n).
\ee
We then define 
\[
H_{\rm f, N} :=  (\bsP_N\otimes \bsP_N)\, H_{\rm f} \, (\bsP_N\otimes \bsP_N), \qquad
H_N := (1 \otimes \bsP_N\otimes \bsP_N) \,H\, (1 \otimes \bsP_N\otimes \bsP_N).
\]
\begin{proposition} \label{P_n is spectral discretization}
$H_{\rm f,N}$ and  $H_N$ are spectral discretizations of $H_{\rm f}$ and $H$. 
\end{proposition}
The proof will be given in Section \ref{proofs}.


\subsection{Main results}

Our main results are about the partition function $Z(\beta)$ of the oscillator embedded in the photon field, and about
its inverse Laplace transform, the effective measure of states.  
It is convenient to introduce dimensionless entities. The  strength of the coupling is given 
by an angle $\varphi$ such that 
\be\label{cosphi}
\mbox{sin}\varphi:=\frac{\eta}{3\,\omega_{CF}}\,\alpha=\frac{\eta}{3mc^3}\,e^2
\ee
holds, where
$\alpha = \frac{e^2}{\hbar c}$ is  the Sommerfeld fine structure constant and $\omega_{CF}=\frac{mc^2}{\hbar}$  the Compton frequency. As to the cutoff $C$ we introduce  $\gamma:=cC/\eta$, and  $\beta=\frac{\hbar}{\rm{k} T}$ is replaced by the variable
$\rho:=\frac{\eta\beta}{2\pi}$.
\begin{theorem} \label{Z(beta) UV regular}
$\Tr \e{- \frac{\beta}{\hbar} H_N}$ and $\Tr \e{- \frac{\beta}{\hbar} H_{\rm f, N}}$ exist for each $N \in \bbN$ and each 
$\beta > 0$. Moreover, $Z(\beta;\gamma) := \lim_{N \to \infty} (\Tr \e{- \frac{\beta}{\hbar} H_N}) / ( \Tr \e{- \frac
{\beta}{\hbar} H_{\rm f,N}})$ exists for each $\beta >0$, and 
\be \label{UVreg}
Z(\beta;\gamma) = 
\left[2\pi\rho\; \emph{e}^{ -2 \rho \ln (1+\gamma) \sin\varphi}
\prod_{l=1}^{\infty}\left(1+\frac{\rho^2}{l^2}+\frac{4}{\pi}\frac{\rho}{l} \sin (\varphi) \arctan \left(\gamma\frac
{\rho}{l}\right)\right) \right] ^{-3}.
\ee
\end{theorem}
Formula \eqref{UVreg} already appears in \cite{CK87}, where it is derived non-rigorously using path integrals. 
We present a conceptually simpler, and mathematically rigorous, proof of Theorem \ref{Z(beta) UV regular} in 
Sections \ref{sec trace formula} -- \ref{sec continuum limit}. 

Removing the ultraviolet cutoff in \eqref{UVreg} results in a significantly simpler expression.
\begin{theorem} \label{Z(beta) UV singular}
For all $\beta>0$, $Z(\beta) := \lim_{\gamma \to \infty} Z(\beta;\gamma)$ exists, namely 
\be \label{UVsing}
Z(\beta) = \left[\frac{\rho}{2\pi} \e{-2 \rho \ln (\rho) \sin \varphi} \left\vert\Gamma(\ii \rho\,\e{-\ii \varphi})\right\vert ^2   \right] ^3.
\ee
Moreover, $\beta \mapsto Z(\beta)$ is the Laplace transform of a positive measure. 
\end{theorem}
The statements of Theorem \ref{Z(beta) UV singular} also appear in \cite{CK87}. But the derivation of 
\eqref{UVsing} is not fully rigorous there and the proof of the final statement contains a small error.
Both results are improved in  Sections
\ref{sec remove uv cutoff}  and \ref{sec complete monotonicity}, respectively.

We now turn to the effective measure of states $\mu_\varphi$ of the charged 
oscillator, i.e.\ the Laplace inverse of $Z_\varphi$. 
We display the dependence on the strength of the coupling by the 
suffix $\varphi \in [0,\frac{\pi}{2}]$. For $\varphi =0$, we are in the 
case of zero coupling, $Z_0(\beta)=(2\sinh\frac{\beta\eta}{2})^{-3}$ is the partition function of the oscillator, and
\be
\mu_0= \sum_{j=0}^\infty 
\left(
\begin{matrix}
j+2\\
2
\end{matrix}
\right)
\delta_{(j+3/2)\eta}.
\ee
By (\ref{pms}), $\mu_\varphi^{(\beta)}=\frac{1}{Z(\beta)}\e{-\beta(\cdot)}\mu_\varphi$ is the probability 
distribution of the energy  at the temperature $T=\frac{\hbar}{\rm k \beta}$. An easy first step is the following 
convergence result. 
\begin{proposition}\label{cont}
As $\varphi \to 0$,  convergence
$\mu_\varphi^{(\beta)} \to \mu_0^{(\beta)}$ and $\mu_\varphi \to \mu_0$
holds in the sense of the vague topology.
\end{proposition} 

\begin{proof}
One finds $\hat{\mu}^{(\beta)}_\varphi(t) =Z_\varphi(\beta +\ii t)/Z_\varphi(\beta)$ for the Fourier transform of $\mu^
{(\beta)}_\varphi$ at $t\in \bbR$.   Since $Z_\varphi(\beta +\ii t)/Z_\varphi(\beta)\to Z_0(\beta +\ii t)/Z_0(\beta)= \hat{\mu}
_0^{(\beta)}(t)$ as $\varphi\to 0$, the vague convergence of $(\mu_\varphi^{(\beta)})_\varphi$  follows from the 
continuity theorem. This implies 
the vague convergence of $(\mu_\varphi)_\varphi$ simply noting  that  if $\psi$ is a continuous function on $\bbR$ with
compact support then so is $\operatorname{exp}(\beta\, (\cdot))\psi$.
\end{proof}
The foregoing result implies that the probability mass of $\mu_\varphi$ concentrates near the lines at $(j+3/2)\eta$ and vanishes in between as the coupling strength tends to zero. This is exactly what is expected. In the following  we study the structure of $\mu_\varphi$ more closely. We define 
\be
\omega_\varphi:= \frac{3}{\pi}
\Big(\sin \varphi + (\frac{\pi}{2}-\varphi) \cos\varphi\Big)\eta,
\ee
and write $\delta_{\omega_\varphi}$ for the Dirac measure at $\omega_\varphi$, and 
$\phi \lambda^1$ for the Lebesgue measure on $\bbR$ with density $\phi$. 

\begin{theorem} \label{ems structure 1} \hspace*{\fill} 
\begin{itemize}
\item[i)]
There exists a  function $\phi:\bbR\to\bbR$ being zero for $\omega< \omega_\varphi$ and positive continuous  for $\omega\ge \omega_\varphi$ with  $\phi(\omega_\varphi) = \frac{\pi}{\eta} \sin \varphi$ such that
\[
\mu_\varphi = \delta_{\omega_\varphi}  + \phi \lambda^1.
\]
\item[ii)]
 $\phi$ is real analytic for $\omega > \omega_\varphi$, and extends to an analytic function on $\bbC$ up to singularities at
\[
p_j=\omega_\varphi +  j\eta \e{- \ii\varphi}\mbox{\;and\; } \overline{p_j}=\omega_\varphi +  j\eta \e{ \ii\varphi},\qquad j \in \bbZ \setminus \{0\},
\]
and cuts along 
$z= \omega_\varphi + s\e{\pm \ii\varphi}$ for $s \in \bbR$, $|s| \geq \eta$. 
\end{itemize}
\end{theorem}

So the zero-point frequency $\frac{3}{2}\eta$ of the  oscillator is shifted down to $\omega_\varphi$ when the oscillator 
couples to the radiation field, and 
$\phi$ is the \textbf{effective density of states} of the charged oscillator. Part i) above is already shown in 
\cite{CK87}, but  in Section \ref{sec complex structure 1}  an elementary proof is given. According to i)
there is a stable ground state at the energy $\hbar \omega_\varphi$, while 
at the same place the absolutely continuous part of the spectrum jumps to $\frac{\pi}{\eta} \sin \varphi$. 
Hence the ground state is no longer isolated. It also implies that there is no other stable and no singular part of the spectrum. Part ii) goes beyond these 
qualitative statements, and already shows the existence of \textbf{resonances} (cf.\ Definition \ref{reso}) at 
$p_j, \,j \in \bbN$. Its proof will be given in Section \ref{sec complex structure 1}. However, it still contains
no information on the nature of the singularities; in order to observe unadulterated \textbf{Lorentz profiles} on the real line, these need 
to be simple poles. This, and more, is true. To formulate the corresponding result, let us define
\[
\ell_j(z):= \frac{-1}{2 \pi \ii (z - p_j)} + \frac{1}{2 \pi \ii (z - \overline{p_j})},
\]
for $z \in \bbC$. When restricted to the real line,
\[
\ell_j(\omega)  = \frac{1}{\pi} \frac{j \eta \sin \varphi}
{(\omega-\omega_\varphi - j \eta \cos \varphi)^2 + (j \eta \sin \varphi)^2}
\]
is a Lorentz profile with total mass $\int \ell_j \, \dd \omega = 1$.

\begin{theorem} \label{ems structure 2}
Let $N \in \bbN$ and $z=\omega_\varphi +s\e{\ii \chi}$ for \,$0< s < N\eta$ and $\vert \chi\vert <\frac{\pi}{2}$.  Then
\[
\phi(z) = \left\{ \begin{array}{ll}
\displaystyle \phantom{-} \sum_{j=1}^N \binom{j+2}{2} \ell_j(z) + h_N(z) & \text{if } \vert \chi\vert  < \varphi,\\[4mm]
\displaystyle  - \sum_{j=1}^3  (-1)^{j} \binom{3}{j} \ell_j(z) + \tilde{h}_N(z) & \text{if } \vert \chi\vert > \varphi.
\end{array} \right.
\]
Above, $h_N$ and $\tilde{h}_N$ are analytic up to singularities at $p_j,\,\overline{p_j}$ and cuts along $z = \omega_\varphi +s \e{\pm \ii \varphi}$, $s \geq \eta$.
Moreover, there exist a constant  $C$ depending on $N$  such that $h_N$ (and similarly $\tilde{h}_N$) satisfies
$ |h_N(z)| \leq C(1+ \vert \ln (\varphi-\vert\chi\vert)\vert^{\eta N})$ for $\vert\chi\vert<\varphi$.
\end{theorem}
The proof of this theorem is given in Section \ref{sec complex structure 2}. The theorem shows that inside the cone $\vert \chi\vert <\varphi$,
the singularities of $\phi$ look like resonance poles of first order,
but their overall structure is more complicated as is apparent from the behaviour for $\vert\chi\vert >\varphi$. Nevertheless, 
they lead to Lorentz profiles $\binom{j+2}{2} \ell_j$ on the real axis. The total mass of the $j$-th 
Lorentz profile thus corresponds precisely to the multiplicity of the $j$-th eigenvalue of the uncoupled harmonic 
oscillator, i.e.\ the mass of the $j$-th delta peak of $\mu_0$, for all $\varphi$. 
Moreover, with Proposition \ref{cont} it follows that 
$h_N \, \dd \omega$ converges vaguely to zero as $\varphi \to 0$. 

To summarize, the following complete picture of $\mu_\varphi$ emerged: There is a stable ground state 
(Dirac peak of strength one) at $\hbar \omega_\varphi$. Above, there is absolutely continuous spectrum, consisting of
Lorentz profiles with widths 
\be\label{hw}
\gamma_{ j}:=2j \eta \operatorname{sin}\varphi=j \eta\frac{2\alpha}{3}\left(\frac{\hbar \eta}{mc^2}\right),
\ee
centered at  
\be\label{center}
\omega_{ j}:= \omega_\varphi +j\eta \operatorname{cos}\varphi= \omega_\varphi+j\eta\left[1-\frac{\alpha^2}
{18}\left(\frac{\hbar \eta}{mc^2}\right)^2+\dots\right],
\ee
and with total mass equal to that of the corresponding unperturbed state.

\bfig
\centerline{a) \includegraphics[width=70mm]{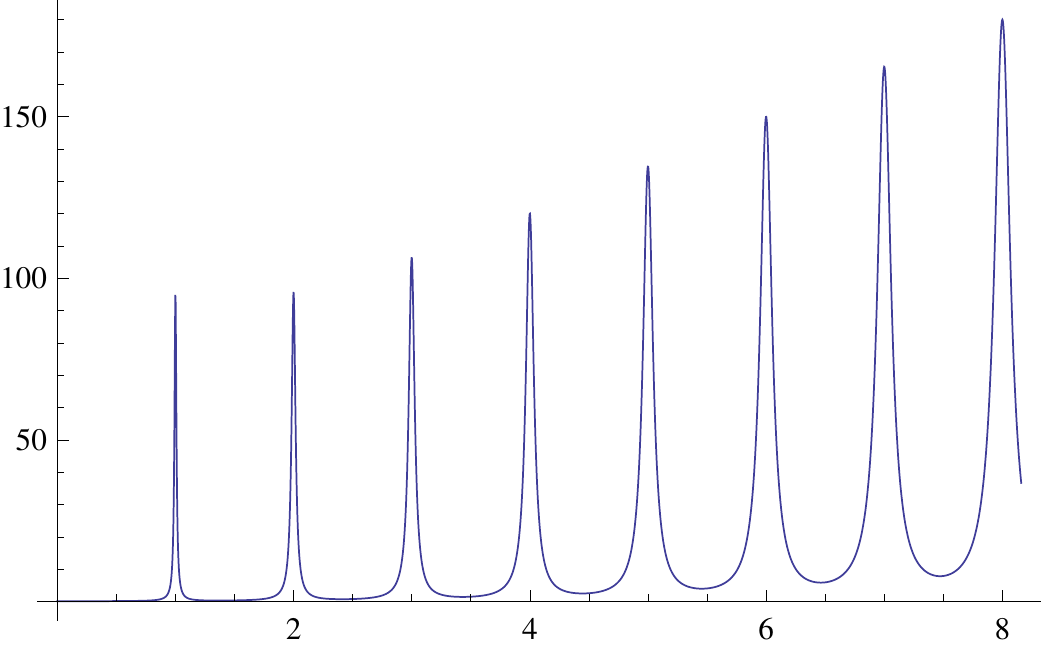}\, b) \includegraphics[width=70mm]{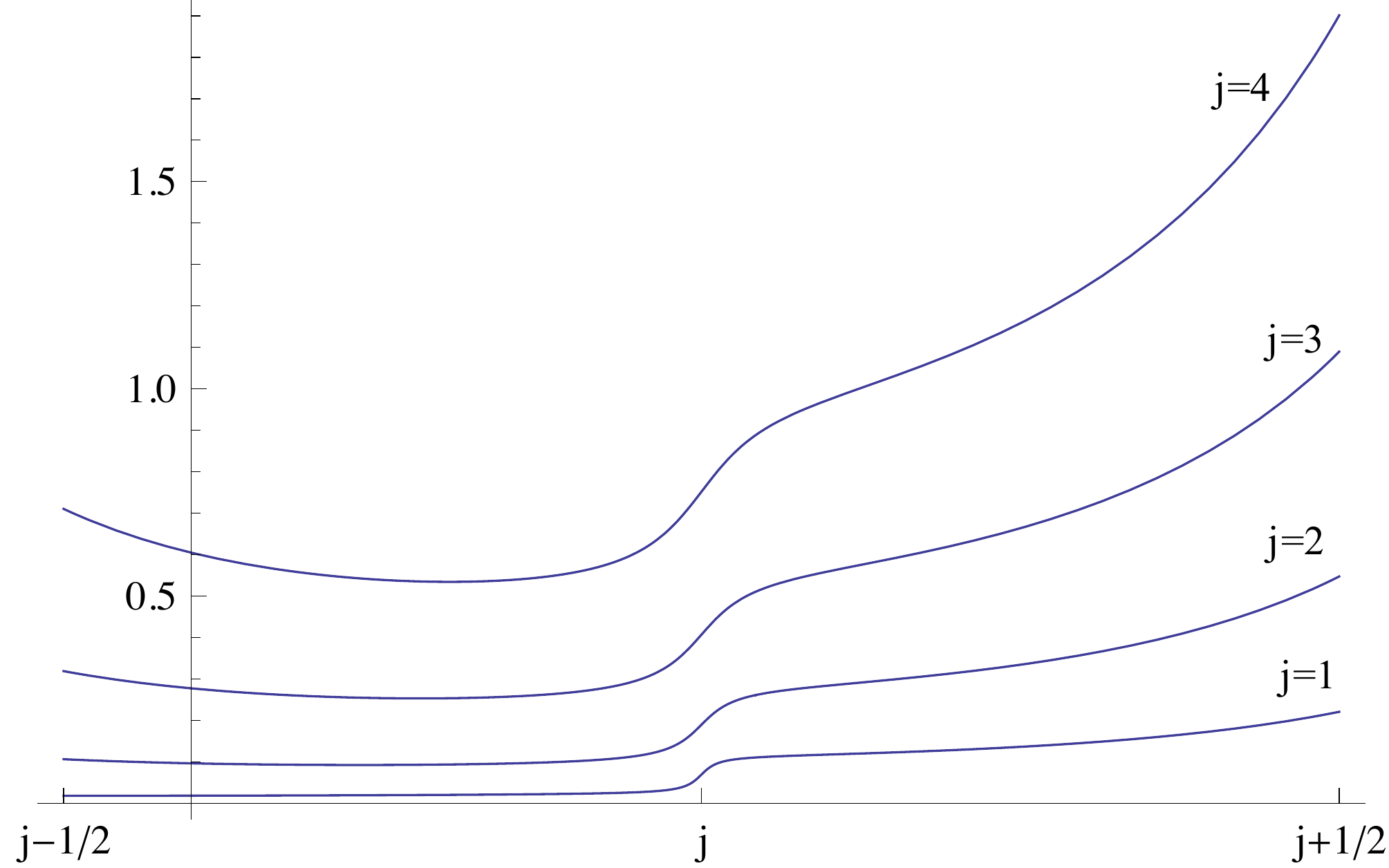}}
\caption{a) shows the effective density of states, computed by summing up the first 15 terms of \eqref{dens}, 
for $\varphi=1/100$. 
All terms of \eqref{dens} are positive, so the numerical approximation is very good.
b) shows the difference 
between the effective density of states and its approximation by the appropriate Lorentz-profiles 
$(j+1)(j+2) \ell_j/2$, in the intervals $[j-1/2,j
+1/2]$, for $j=1,2,3,4$. 
In both graphs, the units of energy on the $x$-axis are chosen such that the eigenvalues of the uncoupled 
harmonic oscillator occur at the elements of $\bbN_0$.}  
\label{fig 1}
\efig

Figure \ref{fig 1} illustrates the 
excellent approximation to the effective density of states obtained by the above description, even for the 
relatively large value $\varphi = 1/100$; the physical value for $\varphi$ is about $\alpha^{-3}\approx 0,5\times 10^{-6}$.

 With respect to a frequency scale starting at $\omega_\varphi$ the peaks  at $\omega_{ j}$ for $j=1,2,3,\dots$  
represent the \textbf{natural lines} of the charged oscillator due to transitions to the ground state (Lyman series). 
According 
to (\ref{center}) the distance of two subsequent levels is $\eta\operatorname{cos}\varphi$, 
which is smaller than $\eta$ in the unperturbed case. The radiative shift (Lamb shift) 
$\eta -\eta\operatorname{cos}\varphi$ of the 
first excited state  is of 
second order in $\alpha$.

For an observation of the natural  lines the charged oscillator may be  brought into contact with  black body radiation  at the temperature $T=\frac{\hbar}{\rm k \beta}$. This causes  additional negative shifts of the excited states depending on $T$ which  have to be taken into account.
They  are determined by the positions of  the peaks of $\mu_\varphi^{(\beta)}$. For $2T>(\hbar/ \rm k)\gamma_{ j}$ one readily finds $-\beta \gamma^2_{ j}/4 <
\Delta\, \omega_{ j}<-\beta \gamma^2_{ j}/8$.

Our results are in  accordance with calculations  within QED time--dependent 
perturbation theory  up to first order in the 
fine-structure constant $\alpha$ (cf. e.g. \cite[sec. 5]{CKS93}). 
The calculations on the radiative cascade of a harmonic oscillator in \cite[Exercises 15]{CDG98} confirm also 
that the shifted levels of the oscillator remain equidistant.

\vspace{1mm}
\noindent \textbf{Remark.} In our analysis, we have left out the self-interaction term of the field ($A^2$)-term). At 
least Theorem \ref{Z(beta) UV regular} can be achieved without this simplification. 
Actually, in the G\"oppert-Mayer description  
this term is absorbed by the coupling, compare \cite[ IV (B.34)]{CDG87} or  \cite[(13.124)]{Sp04} and \cite[V.15 eq. 
(15.1),\,(15.2)]{Si79}. One is left again with an oscillator model which can be treated in the same way as we do here.
The analogous formula to \eqref{UVreg}, without adding any extra renormalization term, is
\[
Z(\beta;\gamma) =
\left[2\pi\rho\; 
\prod_{l=1}^{\infty}\left(1+\frac{\rho^2}{l^2}+\frac{4}{\pi}\sin\varphi\;\left(\gamma - \frac{l}{\rho}\arctan\left(\gamma
\frac{\rho}{l}\right)\right)\right)\right]^{-3}.
\]
Again the limit $\gamma \to \infty$ does not exist, but this time we do not know how to suitably renormalize  this system.

On the other hand, from the well-tested perturbation theory in QED one
knows that for many problems concerning the interaction of bound electrons
with radiation (e.g. spontaneous and stimulated emission and absorption) the
linear part $-epA/m$ of the interaction Hamiltonian yields the main effects
(see e.g. \cite{LM}).



\section{Discretisation of the photon field} \label{proofs}

Let us now give the proof of Proposition \ref{P_n is spectral discretization}, and at the same time prepare the proof of Theorem \ref{Z(beta) UV regular}. In effect, what we give is the reverse procedure to the one that, in the early days of QED, led from the description of 
the photon field as a collection of independent harmonic oscillators to the Fock space description. So in principle, what we are going to present is well known. However, we are not aware of any place where it is done carefully and in a mathematically satisfactory way, and thus present it in some detail. 

Let us start with some properties of the projections $P_N$ introduced in (\ref{P_N}).
For locally integrable $f: \bbR^3 \to \bbC$, and given $j \in J_N$, we write 
\[
\bar f_j:=\frac{1}{V}\langle 1_{\Lambda_j},f \rangle, 
\]
for the average value of $f$ on $\Lambda_j$. Thus, if $f \in L^2$ we have 
\[
P_N f (x)= \sum_{j \in J_N} \bar f_j 1_{\Lambda_j}(x), \quad x\in \bbR^3.
\]

\begin{lemma} \label{L^2 convergence}
Assume $f \in L^2(\bbR^3)$, $F \in \caF$. Then we have: 
\begin{itemize}
\item[(i)] As $N \to \infty$, $P_N f \to f$ in $L^2(\bbR^3)$, and $\bsP_N F \to F$ in $\caF$.
\item[(ii)] If $f$ is continuous on an open set $B$, then $P_N f(x) \to f(x)$ for every $x \in B$.
\end{itemize}
\end{lemma}

\begin{proof}
As to (ii), assume without restriction that $f$ is real-valued. Choose $x_0 \in B$. For all $N$ large enough, we have $\overline{\Lambda}_{j_0} \subset B$ for the unique $j_0$ such that $x_0 \in \Lambda_{j_0}$. By continuity there are $x_N,\,x'_N$ in $\overline{\Lambda}_{j_0}$ satisfying $f(x_N)=\inf f\vert_{\overline{\Lambda}_{j_0}}$ and 
$f(x'_N)=\sup f\vert_{\overline{\Lambda}_{j_0}}$. Then 
$P_Nf(x_0)=\bar f_{j_0}\in [f(x_N),f(x'_N)]$, whence the assertion.

Now (i) is proved in the standard fashion.  Choose $f$ to be the indicator function of a  bounded measurable  set $B$  at first. Then for each $\eps > 0$ 
there is a continuous functions $g$ with compact support 
such that $\norm[L^2]{g - 1_B} < \eps$. Now 
\[ 
\norm{P_N 1_B - 1_B} \leq   \norm{P_N (g - 1_B)} + \norm{P_N g -g} + \norm{g - 1_B}.
\]
The middle term converges to zero by (ii) and dominated convergence, and so the left hand side  is bounded by $2 \eps$ for all large $N$. This proves the claim for indicators of  bounded measurable sets.
As the latter are total in  $L^2$, this proves also the general result  about $ (P_N)_N$.
Finally, the result about $(\bsP_N)_N$ follows from the fact that second quantisation preserves strong convergence. 
\end{proof}

It can be checked easily that 
$\bsP_N \caF = \caF( P_N L^2(\bbR^3))$, the latter being the Fock space over the image of $P_N$. 
Since $\{h_j:\, j \in J_N\}$ is a ortonormal basis of $P_N L^2(\bbR^3)$, we conclude that $\{ h_{j_1} \hat \otimes  \ldots \hat \otimes h_{j_n} : {j_k} \in J_N, \,n \in \bbN_0 \}$ is a orthonormal basis of $\bsP_N \caF$. These basis elements are eigenfunctions of 
$\bsP_N H_{\rm f,\sigma} \bsP_N$, as 
\[
\bsP_N H_{\rm f, \sigma} \, \, h_{j_1} \hat \otimes  \ldots \hat \otimes h_{j_n}  = \left( \hbar \sum_{i=1}^n \bar \omega_{j_i}\right) \, h_{j_1} \hat \otimes  \ldots \hat \otimes h_{j_n}.
\]
This allows us to transfer the operator $H_{\rm f,N}$ unitarily to an $L^2$ space. For $n \in \N$ we write  for any positive mass parameter $\mu$ 
\[
\psi_n(x) = \sqrt{\frac{1}{2^n n!}} \left( \frac{\mu \bar \omega_j}{\pi \hbar}\right)^{1/4} 
\e{-\frac{\mu \bar \omega_j}{2 \hbar} x^2} H_n\left( \sqrt{\frac{\mu \bar\omega_j}{\hbar}} x \right),
\]
for the $n$-th Hermite function on $\bbR$. For  $j \in J_N$ and  $\bsj := (j_1, \ldots, j_n) \in J_N^n,\; n\in \bbN$  let $n(j,\bsj)$ denote  the number of occurrences of the value 
$j$ in  $\bsj$; for $n=0$ set $\bsj :=\emptyset,\, h_\emptyset :=1\in \bbC$ and  $n(j,\emptyset) :=0$ for all $j\in J_N$. We put 
\[
h_{\bsj} := h_{j_1} \hat \otimes \ldots \hat \otimes h_{j_n} \in \caF^{(n)} \text{ and } 
\bspsi_{\bsj}:=\bigotimes_{j\in J_N}\psi_{n(j,\bsj)} \,\in L^2(\bbR^{J_N}).
\]
Due to symmetrisation, $h_{\bsj}$ actually only depends on the occupation numbers $n(j,\bsj),\,j\in J_N$, and thus it is easy to check
that the asignment  $\theta(h_{\bsj}):=  \bspsi_{\bsj}$ defines an isomorphism from $P_N \caF$ onto $L^2(\bbR^{J_N})$, 
and that 
\be \label{free field oscillator}
\theta \bsP_N H_{\rm f,\sigma} \theta^{-1} = \sum_{j \in J_N}  \left( - \frac{\hbar^2}{2 \mu} \partial_{x_{\sigma j}}^2 + 
\frac{\mu}{2} \bar \omega_j^2 x_{\sigma j}^2 - \frac{\hbar \bar\omega_j}{2} \right),
\ee
where $x_{1 j}$ and $x_{2 j}$ are independent variables. Due to the relation $2 x H_n(x) = H_{n+1}(x) + 2 n H_{n-1}(x)$ we have 
\[
\sqrt{n+1} \psi_{n+1}(x) + \sqrt{n} \psi_{n-1}(x) = \sqrt{\frac{2 \mu \bar \omega_j}{\hbar}} x \psi_n(x),
\]
and using the definitions of $a^\ast(g)$ and $a(g)$, it is tedious but straightforward to check that
\[
\theta \bsP_N H_{\rm I,\sigma} \theta^{-1} = \sqrt{\frac{V\mu}{2\pi^2}} \frac{e}{m} \sum_{j \in J_N }
\sqrt{ \bar \omega_j}(\bar y_{\sigma j}\cdot q) x_{\sigma j} \quad \text{with }\; y_\sigma  := u_\sigma \chi_c / \sqrt{\omega}.
\]
Thus we find that $H_N$ is unitarily equivalent to 
\be \label{oscillator}
\tilde H_N = H_{\rm p} + H_{\rm r, N} + \sum_{j \in J_N} \sum_{\sigma = 1,2} \left( - \frac{\hbar^2}{2\mu} \partial_{x_{\sigma j}}^2 + 
\frac{\mu}{2} \bar \omega_j^2 x_{\sigma j}^2 - \frac{\hbar \bar\omega_j}{2} + \sqrt{\frac{V\mu}{2\pi^2}} \frac{e}{m}  
\sqrt{ \bar \omega_j}(\bar y_{\sigma j}\cdot q) x_{\sigma j}\right),
\ee
acting in $L^2(\bbR^{3 + 2 J_N})$.  Here we discretized  also the renormalisztion terms
$$
H_{\rm r ,N} = \sum_{j \in J_N} \sum_{\sigma = 1,2}V\frac{e^2 \eta^2}{2 \pi^2 m}\; \frac{1}{\bar \omega_j}(\bar y_{\sigma j} \cdot q)^2\;-\,\sum_{j \in J_N}V\hbar \frac{e^2\eta^2}{4\pi^2 m}\,\frac{(\overline{\chi_c})_j}{\bar \omega_j^2(\bar \omega_j+\eta)}
$$
displaying the contribution to the effective oscillator mass and energy shift by each single photon.

Since $\tilde H_N$ is the Hamiltonian of a finite dimensional quantum oscillator, its spectrum is discrete and the eigenvalues are of finite multiplicity. Using Lemma \ref{L^2 convergence}, it is immediate to check (i) and (ii) of Definition \ref{spec discret} for $\rm1\otimes \bsP_N\otimes \bsP_N$, 
and Proposition \ref{P_n is spectral discretization} is proved. 

\section{The partition function} \label{proofs2}
In this section we will prove Theorems \ref{Z(beta) UV regular} and \ref{Z(beta) UV singular}.
We start with a general, useful observation. 

\subsection{A trace formula for a quantum oscillator} \label{sec trace formula}
Consider an operator 
\be \label{generalosc}
H=\sum_{i=1}^n \frac{-\hbar^2}{2m_i}\,\partial^2_i + \frac{1}{2}\sum_{i,Úi'=1}^nA_{ii'}\,X_iX_{i'}
\ee
with strictly positive real $n\times n$-matrix $A=(A_{ii'})$.
$H$ acts in $L^2(\bbR^{n})$, $\partial_i$ are derivative operators with respect to $x_i$, and $X_i$ are multiplication operators with $x_i$. We assume that the mass parameters $m_i$ are all positive, and we define the diagonal  mass matrix $M:=\mbox{diag}(m_1,\dots ,m_n)$. 

\begin{lemma}\label{traceformula}
Let $\lambda_1,\dots ,\lambda_n$ be the eigenvalues of
$M^{-1/2}\,A\,M^{-1/2}$.
Then
$$
\Tr \e{- \frac{\beta}{\hbar} H}=\prod_{i=1}^n\left(2\sinh(\frac{\beta}{2}\sqrt{\lambda_i})\right)^{-1}=\det \left(2\,\sinh\left(\frac{\beta}{2}(M^{-1/2}A\,M^{-1/2})^{\frac{1}{2}}\right)\right)^{-1}=
$$
$$
\beta^{-n}\det (M^{-1/2}A\,M^{-1/2})^{-\frac{1}{2}}\prod_{l=1}^\infty\det \left(I_n+(\frac{\beta}{2\pi l})^2M^{-1/2}A\,M^{-1/2}\right)^{-1}.
$$
\end{lemma}
\noindent{\it Proof.} We are going to show that $H$ is unitarily equivalent to
\be \label{unit equiv H}
\tilde H := \sum_{i=1}^n\left( \frac{-\hbar^2}{2m_0}\,\partial^2_i +\frac{m_0}{2}\lambda_i\,X^2_i\right)
\ee
with $m_0>0$ any mass parameter. Then by Mehler's formula the first equality holds. The second equality is obvious.  The last equality holds true due to the factorization
$$
\sinh z = z\prod_{l=1}^\infty\left(1+(\frac{z}{\pi l})^2\right) 
$$
valid for all $z \in \bbC$. To show \eqref{unit equiv H}, we note that any invertible real $n\times n$-matrix $S$ gives rise to a unitary transformation $U$ on  $L^2(\bbR^{n})$ by $Uf:=\vert \mbox{det}S\vert^{\frac{1}{2}}\,f \circ S$. We choose $S$ in the following way: Since $M^{-1/2} A M^{-1/2}$ is real symmetric, there exists an orthogonal matrix $O$ such that $OM^{-\frac{1}{2}}AM^{-\frac{1}{2}}O^T =\mbox{diag}(\lambda_1,\dots ,\lambda_n)$. We put $S:=O(\frac{1}{m_0}M)^{\frac{1}{2}}$.
Then $S$ satisfies
\[
SM^{-1}S^T=\frac{1}{m_0}I_n,   \quad  \quad S^{-1T}AS^{-1}=m_0\,\mbox{diag}(\lambda_1,\dots ,\lambda_n),
\]
and from these relations $\tilde H=U^{-1}\,H\,U$ follows straightforwardly. \hfill{$\Box$}

\subsection{The Partition Functions for $H_{{\rm f},N}$ and $H_N$}

We will now use Lemma \ref{traceformula} in order to compute the partition functions for the discretized systems. Note that from \eqref{free field oscillator} and Lemma \ref{traceformula} we immediately get 
\be \label{tradisfie}
\Tr \e{- \frac{\beta}{\hbar} H_{{\rm f}, N}} = \e{\beta \sum_{j\in J_N}\bar \omega_j} \prod_{j\in J_N}\left(2\sinh(\frac{\beta}{2}\bar \omega_j)\right)^{-2}.
\ee
Now by \eqref{oscillator}, $\tilde H_N$ has the form \eqref{generalosc} up to the constant term 
\[
-\hbar\sum_{j \in J_N}\left(\bar \omega_j+V \frac{e^2\eta^2}{4\pi^2 m}\,\frac{(\overline{\chi}_c)_j}{\bar \omega_j^2(\bar \omega_j+\eta)}\right).\]
Precisely, the mass matrix $M$ is the diagonal matrix of size $2|J_N| + 3$, 
with $\frac{1}{m \eta^2}$ as the first three diagonal elements, and $\mu$ as the remaining ones. 
$A$ is a block matrix of the form 
\be \label{blockmatrix}
A = \left(  
\begin{matrix}
B 		& B_1 		& B_2	 	& \cdots 	& B_{|J_N|} \\
B_1^T 	& s_1 I_2  	& 0 		 	& \cdots 	& 0		  \\
B_2^T  & 0 			& s_2 I_2 	& \ddots 	& \vdots  \\
\vdots  & \vdots		&  \ddots	& \ddots		& 0 		  \\
B_{|J_N|}^T  & 0		& \cdots		& 0		& s_{|J_N|} I_2
\end{matrix}
\right)
\ee
with $B = \frac{1}{m} I_3$ ($I_j$ denotes the $j$-dimensional unit matrix), $s_j = \mu \bar \omega_{j}^2$
for $1 \leq j \leq |J_N|$, and $B_j = (v_{1j},v_{2j}) \in \bbR^{2 \times 3}$, with 
$v_{\sigma j} = \sqrt{\frac{V \mu}{2 \pi^2}} \frac{e}{m} \sqrt{\bar \omega_j} \bar y_{\sigma j}$ 
for $\sigma = 1,2$. It is clear that both of the matrices whose determinant appears in the second line of
Lemma \ref{traceformula} are of the above form, too. We thus need a formula for determinants of matrices of the type \eqref{blockmatrix}. 
\begin{lemma}\label{detformula}
Let $Q$ be a real symmetric  $(2L+3)$-dimensional matrix 
of the form \eqref{blockmatrix},
where $B$ is 3-dimensional and all $ s_j\not=0$. 
Then
$$
\emph{det}\,Q=s_1^2\,\dots \,s_L^2\;\emph{det}\left(\,B-\sum_{j=1}^L\frac{1}{s_j}\,B_j\,B_j^T\,\right).
$$
\end{lemma}
\noindent{\it Proof.} We write $B_L=(a,b)$, where $a$ and $b$ are 3-dimensional column vectors. 
We assume for the moment that $a$ and $b$ are linearly independent. At the end this assumption can be dropped by continuity. Define a $3\times 3$-matrix $X$ by $X^T:=(x,y,z)$ with $x:=a\times b, \; y:=b\times x$ and $z:=x\times a$. Then $XB_L=\vert a\times b\vert^2(e_2,e_3)$ with $I_3=(e_1,e_2,e_3)$ and $X^{-1}=\vert a\times b\vert^{-2}(x,a,b)$ hold.
The following operations do not change det$\,Q$.

First pass to the equivalent matrix  diag$(X,I_{2L})\,Q$\,diag$(X^{-1},I_{2L})$. Then the last block in the first block line becomes zero adding  $-\vert a\times b\vert^2/s_L$ times the last two lines to the second and third line of the matrix. Expanding now the determinant along the last two columns one obtains the factor $s_L^2$ times  the determinant of a $(2L+1)$-dimensional matrix $Q'$. One finds that diag$(X^{-1},I_{2L-2})\,Q'$\,diag$(X,I_{2L-2})$ arises from $Q$ by canceling the last block line and last block column and replacing $B$ with $B-\frac{1}{s_L}B_LB_L^T$.
Iterating these operations the result follows.\hfill{$\Box$}

Thus by  elementary computations, Lemma \ref{traceformula} and Lemma \ref{detformula} yield the following formula for the ratio of the partition functions of the discretised systems:
\be \label{tradissys}
\begin{split}
Z_N(\beta) := &  
\frac{\Tr \e{- \frac{\beta}{\hbar} H_ N}}{ \Tr \e{- \frac{\beta}{\hbar} H_{\rm f, N}}} =
(\beta\eta)^{-3}\e{\beta\sum_{j \in J_N}V \frac{e^2\eta^2}{4\pi^2 m}\,\frac{(\overline{\chi}_c)_j}{\bar \omega_j^2(\bar \omega_j+\eta)}} \times \\
 & \times
\prod_{l=1}^\infty \det \left( (1 + \frac{\eta^2}{\nu_l^2}) \bsI_3 + 
\frac{e^2 \eta^2}{2 \pi^2 m} V \sum_{j \in J_N} \sum_{\sigma=1,2} \frac{\bar y_{\sigma j} \bar y^T_{\sigma j}}{\bar \omega_j (\bar \omega_j^2 + \nu_l^2)} \right)^{-1},
\end{split}
\ee
with
\[
y_\sigma := \chi_c \frac{1}{\sqrt\omega} u_\sigma, \qquad \nu_l = \frac{\pi l}{\beta/2}.
\]

\subsection{The continuum limit} \label{sec continuum limit}
We will now show that in formula \eqref{tradissys}, the limit $N \to \infty$ exists, yielding (\ref{UVreg}). We use the abbreviation
\[
S_{l, N} := \frac{e^2 \eta^2}{2 \pi^2 m} V \sum_{j \in J_N} \sum_{\sigma=1,2} \frac{\bar y_{\sigma j} \bar y^T_{\sigma j}}{\bar \omega_j (\bar \omega_j^2 + \nu_l^2)}.
\]
Let $\lambda$ denote the Lebesgue measure on $\bbR^3$  and put
\be \label{I_n}
S_l := \frac{e^2 \eta^2}{2 \pi^2 m} \sum_{\sigma=1,2} \int \frac{u_{\sigma} u^T_{\sigma}}
{\omega^2 (\omega^2 + \nu_l^2)} \chi_c \,  \dd \lambda.
\ee

\begin{proposition} \label{riemann convergence}
$\lim_{N \to \infty} S_{l,N} = S_l$ for each $l \in \N$.  
\end{proposition}

\begin{proof}
We will use dominated convergence. Put 
\be \label{g_N}
g_{l,N} := \sum_{j \in J_N} \frac{\bar y_{\sigma j} \bar y^T_{\sigma j}}{\bar \omega_j (\bar \omega_j^2 + \nu_l^2)}  1_{\Lambda_j}.
\ee
By the definition of $V$, we have 
\be \label{e1}
S_{l,N} = \frac{e^2 \eta^2}{2 \pi^2 m} \sum_{\sigma=1,2} \int g_{l,N} \dd \lambda.
\ee
The functions $y_\sigma$, which contain the transverse vector fields, can be chosen to be continuous outside a closed  Lebesgue null subset of $\bbR^3$ containing the origin.
Thus applying Lemma \ref{L^2 convergence} (ii) on every factor
on the right hand side of (\ref{g_N}) one deduces that $g_{l,N}$ converges to the integrand in (\ref{I_n}) almost everywhere. 

Now each component of $g_{l,N}$ is bounded by 
\[
\sum_{j \in J_N} \left(\overline{1/\sqrt{\omega}}\right)_j^{\,\,2} \frac{1}{\bar\omega_j \nu_l^2} 1_{\Lambda_j} \leq \sum_{j \in J_N} \left(\overline{1/\sqrt{\omega}}\right)_j^{\,\,4}
\frac{1}{\nu_l^2} 1_{\Lambda_j},
\]
where the inequality is obtained by applying Jensen's inequality to the convex function $\phi: (0,\infty) \to \bbR,\, x \mapsto 1/\sqrt{x}$, yielding $1/\sqrt{\bar\omega_j} \leq (\,\overline{1/\sqrt{\omega}}\,)_j$ for each $j$. Furthermore, we claim that 
\be \label{upperbound}
\left(\overline{1/\sqrt{\omega}}\right)_j 1_{\Lambda_j} \leq 4 \frac{1}{\sqrt{\omega}} 1_{\Lambda_j}
\ee

\noindent Indeed, in the case $0 \in \bar \Lambda_j$, we replace $\Lambda_j$ with the portion of the ball of radius $\sqrt{3}/N$ 
that lies in the same sector. Then $\Lambda_j$ is contained in that set, and 
\[
\sqrt{c}\left(\overline{1/\sqrt{\omega}}\right)_j = \frac{1}{V} \frac{4 \pi}{8} \int_{0}^{\sqrt{3}/N} r^{3/2} \, \dd r =
\frac{3^{5/4}\pi}{5} N^{1/2} \leq \frac{3^{3/2} \pi}{5 \sqrt{r}} \mbox{\quad if }r\le \sqrt{3}/N,
\] 
whence the claim in this case. If $0 \not\in \bar \Lambda_j$, we let $(k/N,l/N,m/N)$ be the corner of $\Lambda_j$ 
closest to the origin, and find 
\[
\sqrt{c}\left(\overline{1/\sqrt{\omega}}\right)_j \leq \frac{\sup_{\Lambda_j} (1/ \sqrt{r})}{\sqrt{r} \inf_{\Lambda_j} (1/\sqrt{r})} = \frac{1}{\sqrt{r}}
\frac{\sqrt{(k+1)^2/N^2 + (l+1)^2/N^2 +(m+1)^2/N^2 }}{\sqrt{k^2/N^2 + l^2/N^2 +m^2/N^2}}.
\]
The last expression is maximal for $(k,l,m) = (0,0,1)$ and takes the value $\sqrt{6}$ there. Thus (\ref{upperbound}) holds. 

Hence every component of $g_{l,N}$ is bounded by $(16/\nu_l)^2\,\omega^{-2}$, which is locally integrable. This finishes the proof.
\end{proof}

A similar, but easier proof shows that 
\be \label{exponent cont lim}
\lim_{N \to \infty} \sum_{j \in J_N} V \frac{e^2\eta^2}{4\pi^2 m}\,\frac{(\overline{\chi}_c)_j}{\bar \omega_j^2(\bar \omega_j+\eta)} = \frac{e^2\eta^2}{4 \pi^2 m} \int \frac{\chi_c}{\omega^2 (\omega + \eta)} \dd\lambda.
\ee

The integrals appearing in \eqref{exponent cont lim} and \eqref{I_n} can be calculated. 
We  have 
\[
\int \frac{\chi_c}{\omega^2 (\omega + \eta)} \dd\lambda 
= \frac{4 \pi}{ c^3} \ln(1 + \tfrac{C c}{\eta}).
\]
Together with the prefactor, and comparing with the definitions of $\sin \varphi$, $\rho$ and $\gamma$, we obtain the 
exponent in \eqref{UVreg}.
For \eqref{I_n}, note first that 
$\sum_{\sigma=1,2} u_\sigma u_\sigma^T = \bsI_3 - \frac{1}{|k|^2} k k^T$ follows from the orthonormality of $u_1$, $u_2$ and $k/|k|$. Thus for each component of $\sum_{\sigma=1,2} \int \frac{u_{\sigma} u^T_{\sigma}}
{\omega^2 (\omega^2 + \nu_l^2)} \chi_c \,  \dd \lambda$, 
the angular part can be calculated, and is found to be equal to $\frac{8 \pi}{3}$. Then 
\[
\sum_{\sigma=1,2} \int \frac{u_{\sigma} u^T_{\sigma}}
{\omega^2 (\omega^2 + \nu_l^2)} \chi_c \,  \dd \lambda = 
\frac{8 \pi}{3 c^3 \nu_l}  \arctan \left( \tfrac{ c C}{ \nu_l} \right) \bsI_3.
\]
Taking the prefactor into account and computing the now trivial determinant gives the $l$-th factor of the infinite product in \eqref{UVreg}.

The final step is to justify the exchange of the infinite product with the limit $N \to \infty$. To this end, we write,
for each $l \in \N$, 
\[
M_{l,N} :=  \frac{\eta^2}{\nu_l^2} \bsI_3 + S_{l,N}, \qquad M_l := \lim_{N \to \infty} M_{l,N} = 
\frac{\eta^2}{\nu_l^2} \bsI_3 + S_{l}.
\]
\begin{lemma}
We have 
\[
\lim_{N \to \infty} \prod_{l=1}^\infty  \det (\bsI_3 + M_{l,N}) =  \prod_{l=1}^\infty  \det (\bsI_3 + M_{l}).
\]
\end{lemma}

\begin{proof}
By the continuity of the exponential function, it will suffice to prove the result for the logarithms. Then 
the quantity of interest is given by 
\[
\ln \left( \prod_{l=1}^\infty \det (\bsI_3 + M_{l,N}) \right) = \sum_{l=1}^\infty
\ln ( \det (\bsI_3 + M_{l,N}) ).
\]
Since by continuity
\[
\lim_{N \to \infty} \ln ( \det (\bsI_3 + M_{l,N}) ) = \ln ( \det (\bsI_3 + M_{l}) )
\]
for each $l \in \N$, the result will follow by dominated convergence.

At the end of the proof of Proposition \ref{riemann convergence} we have seen that $S_{l,N}$ is bounded by 
${\rm constant} / l^2$ uniformly in $N$. Thus so is $\Vert M_{l,N}\Vert$ and hence its three eigenvalues.
It follows that $|\det (\bsI_3 + M_{l,N}) ) - 1| \leq {\rm constant}/ l^2$ and therefore
\[
|\ln ( \det (\bsI_3 + M_{l,N}) )| \leq  {\rm constant}/ l^{2}
\] 
uniformly in $N$.
\end{proof}
We have thus proved Theorem \ref{Z(beta) UV regular}.


\subsection{Removing the ultraviolet limit} \label{sec remove uv cutoff}
We now investigate the limit $\gamma \to \infty$ in \eqref{UVreg}. Clearly, what we need to show 
is the convergence of 
\[
f(\gamma) := -2\sin \varphi \;\rho \ln (1+\gamma) + \sum_{l=1}^{\infty} \ln (1 + h_l(\gamma)),
\]
where 
\[
h_l(\gamma) = \frac{\rho^2}{l^2} + \frac{4}{\pi} \sin \varphi \, \arctan \left( \frac{\gamma \rho}{l}\right).
\]
We first note that 
\begin{align}
&0< h_l(\gamma) < \frac{\rho^2}{l^2} \left( 1+\frac{4}{\pi}\gamma \sin\varphi\right) \label{c}\\
&0< h_l(\gamma)  \nearrow  h_l:= \frac{\rho^2}{l^2} + 2 \frac{\rho}{l}\sin\varphi \text{ \;monotonically for } \gamma \nearrow \infty. \label{d}
\end{align}
By \eqref{c},  $(h_l(\gamma))$ is summable for all $\gamma$, and thus 
\be \label{f decomp}
f(\gamma) = \left( \sum_{l=1}^\infty h_l(\gamma) - 2 \sin\varphi \;\rho \ln (1+\gamma) \right) + \sum_{l=1}^\infty \left( \ln \big(1 + h_l(\gamma) \big) - h_l(\gamma)\right).
\ee
We now use the formula 
\be \label{eulerM}
\mathcal{C}= \lim_{s\to \infty}  \left[ \frac{2}{\pi}\sum_{l=1}^\infty \frac{1}{l}\arctan (\frac{s}{l})\,-\ln s\right]
\ee
for the Euler/Mascheroni constant $\mathcal{C}$. The first bracket above is then equal to 
\[
\sum_{l=1}^\infty \frac{\rho^2}{l^2} 
+ 2 \sin\varphi\; \rho \left( \ln \rho + \ln \left(\tfrac{\gamma}{1+\gamma}\right) + 
\frac{2}{\pi} \sum_{l=1}^\infty \frac{1}{l} \arctan\left(\tfrac{\rho \gamma}{l}\right) - \ln(\rho \gamma) \right),
\]
and by \eqref{eulerM} converges to 
$ 
\sum_{l=1}^\infty \frac{\rho^2}{l^2} + 2 \sin(\varphi) \rho (\caC + \ln(\rho))$.
The last sum in \eqref{f decomp} converges to $\sum_{l=1}^\infty \left( \ln \big(1 + h_l \big) - h_l\right)$ by \eqref{d} and monotone convergence, since $\ln(1+x)-x$ is monotone decreasing. Combining these two results, we obtain 
\[
\begin{split}
f & := \lim_{\gamma \to \infty} f(\gamma) =  2 \sin\varphi \; \rho (\caC + \ln(\rho)) + \sum_{l=1}^\infty\left( \ln \big(1 + h_l \big) - 2  \frac{\rho}{l} \sin\varphi \right)\\
& = 2 \sin\varphi \; \rho (\caC + \ln(\rho)) + \ln \left( \prod_{l=1}^\infty (1+ h_l) \e{- 2 \frac{\rho}{l}} \sin\varphi \right).
\end{split}
\]
Taking into account that $1+h_l=(1+\frac{\rho}{l}i\e{-i\varphi})(1+\frac{\rho}{l}i\e{i\varphi})$, and writing 
$2i \sin\varphi = \e{i \varphi} - \e{- i \varphi}$, we find 
\[
\e{f} = \frac{1}{\rho^2} 
\e{2 \rho \ln(\rho) \sin\varphi} \left| \rho \e{ \caC \rho \,
i\e{-\ii \varphi}} 
\prod_{j=1}^\infty \left( 1 + \frac{\rho}{l} i\e{-\ii \varphi} \right) \e{\frac{\rho}{l} i\e{-\ii \phi}} 
\right|^2.
\]
Using the representation 
\[
\Gamma(z)^{-1}=z\e{\mathcal{C}z}\prod_{l=1}^\infty(1+\frac{z}{l})\e{-\frac{z}{l}} \quad  \text{ for } \operatorname{Re}\,z>0,
\]
for the Gamma function on the right half-plane, we arrive at (\ref{UVsing}).

\section{The effective measure of states}

In this section, we investigate the inverse Laplace transform of $Z(\beta)$. As a first step, we prove the 
second part of Theorem \ref{Z(beta) UV singular}.

\subsection{Complete monotonicity} \label{sec complete monotonicity}
By Bernstein's theorem, a function $f$ on $[0,\infty[$ is the Laplace transform of a positive Borel measure if and only if
it is completely monotone (c.m.), i.e.\ $(-1)^n \partial_x^n f(x) \geq 0$ for all $x \geq 0$. To show that 
$Z(\beta)$ is c.m., we use (\ref{UVsing}) in order to get  
\[
\frac{1}{3}\ln Z(\beta) = \ln \rho - \ln (2 \pi) - 2 \sin \varphi\; \rho \ln \rho + \ln \Gamma(\rho\, i\e{-\ii \varphi})
+ \ln \overline{\Gamma(\rho\, i\e{-\ii \varphi})}.
\]
We now use  Binet`s formula 
\[
\ln \Gamma (z) = \int_{0}^\infty \frac{\e{- t z}}{t} \left( \frac{1}{1 - \e{-t}} - \frac{1}{t} - \frac{1}{2}
\right) \, \dd t + \frac{\ln (2 \pi)}{2} + (z-\tfrac{1}{2}) \ln (z) - z,
\]
valid for $\Re(z) > 0$. Since
$\ln (\bar z) = \overline{\ln(z)}$ for all $z \in \bbC$, we obtain  
\be\label{lnZ}
\begin{split}
\ln Z(\beta)
 =   \int_0^\infty \e{-t \rho}\,g(t)\, \dd t 
\,\; -  \,t_\varphi \,\rho
\end{split}
\ee
with 
\[
g(t):= \frac{6}{t}\, \Re\left( \frac{1}{1 - \e{-\tau}} - \frac{1}{\tau} - \frac{1}{2} \right),\quad \tau: =  i\e{-i\varphi}t,\quad t_\varphi:=6  \big( \sin \varphi+(\frac{\pi}{2}-\varphi) \cos \varphi \big).
\]
Below we will show that $g(t) \geq 0$ for all $t \geq 0$. This immediately implies that the first term 
on the right hand side of \eqref{lnZ} is c.m. Since the product of two c.m.\ functions is c.m.\, also
$\exp(f)$ is c.m.\ if $f$ is. Thus $Z(\beta)$ is c.m.\ as the product of 
$\exp (\int_0^\infty \e{-t\rho} g(t) \, \dd t)$ and the the clearly c.m. function $\exp(-t_\varphi \rho)$. 
The claim  $g(t) \geq 0$ follows from 
\begin{lemma}\label{pos}
For all  $w \in \bbC$ with $\Re(w) > 0$ we have
\[
h(w) := \Re \left( \frac{1}{1 - \e{-w}} - \frac{1}{w} - \frac{1}{2} \right) \geq 0.
\]
Equality holds if an only if $\Re(w) = 0$. 
\end{lemma}
\begin{proof}
We write $w = u + \ii v$.
If  $u = 0$, then $h(\ii v) = \frac{1 - \cos v}{1 - 2 \cos v + 1} - 
\frac{1}{2} = 0$. For $u>0$, a direct calculation yields 
\[
\begin{split}
h(z) &= \frac{(u^2 + v^2) (1 - \e{-2 u}) - 2 u (1 + \e{-2 u} - 2  \e{-u} \cos v)}
{2 (1 - 2 \e{-u} \cos v + \e{-2 u}) (u^2 + v^2)} =\\
& = \frac{(u^2 + v^2) \sinh(u) - 2 u (\cosh(u) -\cos(v))}
{2 (u^2+v^2) (\cosh(u) - \cos(v))}
\end{split}
\]
with  nonnegative denominator. We need to show that the numerator  $r(u,v)$  is positive 
if $u>0$.
By symmetry it suffices to treat $v \geq 0$. 
Take first $v=0$. Then $\partial_u (r(u,0)/u) = u \cosh u - \sinh u$, which at $u=0$ equals $0$. Since 
$\partial_u^2 (r(u,0)/u) = u \sinh u>0$, $u \mapsto r(u,0) > 0$ for $u>0$. Thus the result holds if 
$\partial_v r(u,v) = 2 (v \sinh u -u \sin v)$ is nonnegative   for all $u >0, \,v>0$. But this is clearly true, since $(\sinh u) / u >1 > (\sin v) / v$ for all $u, v \neq 0$.
\end{proof}


\subsection{Analytic continuation} \label{sec complex structure 1}
Here we prove Theorem \ref{ems structure 1}. Let us define
\[
Y(\rho) := Z(\beta) \e{t_\varphi \rho}.
\]
Since $\rho = \frac{\eta \beta}{2 \pi}$, the inverse Laplace transforms of $Y(\rho)$ and $Z(\beta)$ are related by translation and 
scaling. 
The key observation is that, due to (\ref{lnZ}), $Y$ is itself the exponential of a Laplace transform, 
\be\label{part}
Y(\rho)=  \exp \mathfrak{L} g (\rho) =  1 + \sum_{n=1}^\infty \frac
{1}{n!} \mathfrak{L}  g^{\ast n}(\rho).
\ee
Above, $\mathfrak{L}$ denotes the Laplace transformation,  and
\be \label{convol}
h \ast k\,(x) = \int_0^x h(t) k(x-t) \, \dd t
\ee
is the convolution of functions supported on $[0,\infty[$. Since $g$ is bounded on compact intervals and nonnegative, 
the final sum in \eqref{part} converges uniformly on compact intervals. 
Let us define 
\be \label{dens}
\varrho := \sum_{n=1}^\infty  \frac{1}{n!}\, g^{\ast n}.
\ee
By Lemma \ref{pos}, all terms of the above sum are positive. Thus $\mathfrak{L} \varrho = Y-1$ holds 
by monotone convergence, and the uniqueness of Laplace transforms shows 
$\mathfrak{L}^{-1} Y = \delta_0 + \varrho$. 
Moreover, a direct  calculation shows that $g(t) \to \frac{1}{2} \sin \varphi$ as $t \to 0$. Thus the same
holds for $\varrho$, and we have shown the first part of Theorem \ref{ems structure 1}, with 
$\phi(\omega) =  \frac{2 \pi}{\eta} \varrho \big( \tfrac{2 \pi}{\eta} (\omega -\omega_\varphi)\big)$ for $\omega \ge \omega_\varphi$ and zero else.

We now investigate the analytic continuation of $\varrho$. First note that 
the analytic continuation of $g$ is given by
\[
g(z) = \frac{3}{z}\left[ \frac{1}{1-\exp(-\ii\e{-\ii\varphi}z)}+ \frac{1}{1-\exp(\ii\e{\ii\varphi}z)}-\frac{2\sin \varphi}
{z}-1\right].
\]
\begin{lemma} \label{g complex}
$g$ is a meromorphic function on $\bbC$ with simple poles at the zeros of the denominators
\[
q_j = 2\pi j \e{-\ii \varphi} \quad \text{ and } \quad  \overline{q_j} = 2\pi j \e{\ii \varphi}, \qquad j \in \bbZ \setminus \{0\}.
\]
The residues are \,$\operatorname{Res}(g;q_j)=\frac{3\ii}{2\pi j}$, \,$\operatorname{Res}(g;\overline{q_j})=\frac{-3\ii}{2\pi j}
$.
\end{lemma} 
The proof is routine. Note only that we have seen already above 
that the common zero $z=0$ of the denominators actually is a regular point of $g$.

 Let us now consider the analytic continuation of the convolutions. Define  
$\caP :=\{z: \, z=q_j \mbox{ or } z=\overline{q_j} \text{ for all } j \in \bbZ \setminus \{0\}\}$ and
$\caS :=\{z:\,z=s\e{\pm i\varphi} \mbox{ \;for real\, } s \mbox{\, with\, } \vert s\vert\ge 2\pi\}$.
For brevity use $A(\caS,\caP)$ to denote the set of functions $f$ that are analytic on $\bbC \setminus \caS$ 
such that the components of $\caS\setminus\caP$ are cuts for $f$ emanating from the points of $\caP$. Analytic 
functions on $\bbC\setminus \caP$ are elements of  $A(\caS,\caP)$. ---  For any function $f$  and any subset $K$ of $\bbC
$ let $\Vert f\Vert_K$ denote the supremum of $\vert f\vert$ on $K$.

\begin{lemma} \label{anal cont}
Let $h\in A(\caS,\caP)$ and let $k$ be analytic on $\bbC\setminus\caP$. For $z\in\bbC\setminus \caS$ set
\[
F(z) := \int_0^z h(\zeta) k(z-\zeta) \, \dd \zeta
\]
integrating along the straight line joining $0$ and $z$. Then $F\in A(\caS,\caP)$  and $F$ is an analytic continuation of 
$h \ast k$ as defined in \eqref{convol}.
Let $K$ be any compact subset of $\bbC\setminus \caS$ and let $\tilde{K}$ denote the union of all straight lines joining 
the origin with some point of $K$. 
Then $\Vert F\Vert_K \le \Vert id\Vert_K\,\Vert h\Vert_{\tilde{K}}\,\Vert k\Vert_{\tilde{K}}$.
\end{lemma}
\begin{proof}
Plainly, $F$ agrees with $h \ast k$ on the real axis  and is differentiable  at all 
$z \in \bbC\setminus \caS$.  Hence F is an analytic continuation of $h \ast k$ on $\bbC\setminus \caS$.  --- We turn to 
the analytic continuation  $\tilde{F}$ of $F$, e.g., from above across  the cut  between the points $q_j$  and $q_{j+1}$ 
with $j\in \bbN$. Let $z\in \caS_1:=\{z:\,z=s\e{-\ii \varphi} \mbox{ \;for  }s\ge 2\pi\}$ be between  $q_j$  and $q_{j+1}$. 
Join the points $0$  and $z$ along 
$\caS_1$ but avoiding $q_l$, respectively $z-q_l$,  for $l=1,\dots,j$, by making a small detour above, respectively 
below, $\caS_1$. By  analogous curves $\gamma_z$ one joins $0$ to $z$  for every $z$  in  the open disk $D$ centered at $
\frac{1}{2}(q_j+q_{j+1})$ with radius $\pi$. Let $\tilde{h}$ be an analytic continuation  of $h$ from above $\caS_1$ 
across the cuts along $\caS_1$. Then
\[
\tilde{F}(z) := \int_{\gamma_z} \tilde{h}(\zeta) k(z-\zeta) \, \dd \zeta
\]
defines an analytic function  on $D$. It  extends $F$, since for $z\in D$ above $\caS_1$ the closed curve composed by $
\gamma_z$ and the straight line from $z$ to $0$ does not contain any singularity of $\zeta\to \tilde{h}(\zeta) k(z-\zeta)
$. --- The remainder is obvious.
\end{proof} 

By the lemma, $g^{\ast n}\in A(\caS, \caP)$ for all $n \in \bbN$. Moreover, $\Vert g^{\ast 
n}\Vert_K\le \Vert \rm{id} \Vert_K^{n-1}\,\Vert g\Vert_{\tilde{K}}^n$. Obviously, a similar estimate holds more generally for 
the analytic continuations $\tilde{F}$ of $F$ on compact $K\subset \operatorname{dom}\tilde{F}$. Hence the series in 
(\ref{dens}) converges uniformly on compact sets implying that the limiting function belongs to $A(\caS, \caP)$.  
This proves the second part of Theorem \ref{ems structure 1} when taking into account that the translation and scaling 
takes $q_j$  into $p_j$. 

Let us comment on the cuts of $F$ from  Lemma \ref{anal cont}. Even if $h$ and $k$ are analytic on 
$\bbC\setminus \caP$ with poles  at points of $\caP$, the convolution $F$ may have cuts. More precisely, e.g., if $z$ 
lies on the cut between $q_j$ and $q_{j+1}$  with $j\in \bbN$ then 
\[
F_+(z)-F_-(z) = -2\pi\,\ii \sum_{l=1}^j\left[k(z-q_l)\,\operatorname{Res}(h;q_l)+h(z-q_l)\,\operatorname{Res}(k;q_l)\right]
\]
where $F_+(z)$ and $F_-(z)$ denote  the limit values of $F$ at $z$ approaching $z$ from above and from below the cut, respectively. This is an immediate consequence of  Residue Theorem integrating the meromorphic function $\zeta\to M_z(\zeta):=h(\zeta)k(z-\zeta)$ along the simply closed curve, which is symmetric with respect to $\caS_1$ and which joins $0$ to $z$ by $\gamma_z$. --- In case of $F=g\ast g$ the above formula yields the jump function 
\[
F_+(z)-F_-(z) = \sum_{l=1}^j\,\frac{6}{l}g(z-2\pi\,l\e{-i\varphi}).
\]


\subsection{Analysis of the singularities} \label{sec complex structure 2}
Here we prove Theorem \ref{ems structure 2}. Again it suffices to analyse $\varrho$ (\ref{dens}) instead of $\phi$ as $\phi(\omega) =  \frac{2 \pi}{\eta} \varrho \big( \tfrac{2 \pi}{\eta} (\omega -\omega_\varphi)\big)$ for $\omega \ge \omega_\varphi$. We refer to the analytic continuations of the convolutions $g^{\ast n}$ and of $\varrho$ defined in Section \ref{sec complex structure 1}. 

First note the following formula for $z =  s \e{\ii \chi}$ with $s>0$ and $\vert \chi \vert<\varphi$

\be \label{integral}
\int_0^z \frac{1}{\zeta -q_j} \frac{1}{z-\zeta -q_k} \, \dd \zeta = 
\frac{-2\pi \ii}{z-q_j-q_k} +\frac{\ln(z-q_j) -\ln q_j +\ln(z-q_k) -\ln q_k}{z-q_j-q_k},
\ee\\

\noindent which follows from the partial fraction expansion $ \frac{1}{\zeta -q_j} \frac{1}{z-\zeta -q_k} = \frac{1}{z-q_j-q_k} \left( \frac{1}{\zeta-q_j} + 
\frac{1}{z-\zeta-q_k}\right)$  evaluating the primitives $\ln(\zeta-q_j)$ and $-\ln(z-\zeta-q_k)$ of 
$(\zeta-q_j)^{-1}$ and $(z-\zeta-q_k)^{-1}$, respectively. Note that the second term in \eqref{integral} is regular at 
$q_j+q_k=q_{j+k}$. Next let us define  recursively the coefficients

\be \label{def cjn}
c_{j1} = \tilde c_{j1} = \frac{3}{j}, \qquad c_{j,n+1} = \sum_{k=1}^{j-1} c_{k1} c_{j-k,n}, 
\quad \tilde c_{j,n+1} = - \sum_{k=1}^{j-1} \tilde c_{k1} \tilde c_{j-k,n}
\ee
for all $j,n \in \bbN$, where the void sums for $j=1$ are zero. Set
\[
s_{jn}(z) := \frac{-c_{jn}}{2 \pi\ii (z-q_j)}, \qquad \tilde s_{jn}(z) :=  \frac{- \tilde c_{jn}}{2 \pi\ii  (z-q_j)}.
\]

\bigskip
In the following, we will concentrate on the forth quadrant  of $\bbC$. Subsequently it will be easy to extend the result to the  right half-plane. Fix $N \in \bbN$ and  let $z =  s \e{\ii \chi}$ with $0\le s\le2\pi N$ and\, $-\frac{\pi}{2} \le \chi \le 0$. 

\begin{proposition} \label{main step prop}
Then
\[
g^{\ast n}(z) = \sum_{j=1}^N s_{jn}(z) + L_n(z) \qquad \text{for } \, \vert\chi\vert   < \varphi,
\] 
\[
g^{\ast n}(z) = \sum_{j=1}^N \tilde s_{jn}(z) + \tilde{L}_n(z) \qquad \text{for }\vert \chi\vert>\varphi
\]
hold with $\vert L_n(z)\vert$ and $\vert \tilde L_n(z)\vert$  bounded  by 
$A^n (1+\big\vert \ln\sin | \varphi - \vert\chi\vert \vert \big\vert^s)$ for some constant $A$.
\end{proposition}
\begin{proof}
We proceed by induction. The statement for $n=1$ follows from Lemma \ref{g complex}; indeed, subtracting from $g$ the  first order poles  at $q_1,\dots,q_N$ leaves a 
bounded function $L_{1}$. Let us now assume that 
$g^{\ast n}$ has the asserted decomposition and consider  the case $\vert \chi\vert <\varphi$.
Then
\be \label{g n+1}
\begin{split}
g^{\ast( n+1)}(z)
&=  \sum_{j,k=1}^N \int_0^z s_{j1}(\zeta) s_{kn}(z-\zeta) \, \dd \zeta \\
& +  \sum_{j=1}^N \int_0^z \left( s_{j1}(\zeta) 
L_n(z-\zeta) + s_{jn}(\zeta) L_1(z-\zeta)\right) \, \dd \zeta .
\end{split}
\ee
By  (\ref{integral}) 
we find 
\be \label{large formula}
  \int_0^z s_{j1}(\zeta)
s_{kn}(z-\zeta) \, \dd \zeta =\frac{-c_{j1}c_{kn}}{2\pi \ii(z-q_{j+k})} + L_{jkn}(z),
\ee
where $L_{jkn}$  is regular at $q_{j+k}$.
The first term above contributes to $s_{j+k,n+1}$. 
We will show below that none of the remaining terms entering $g^{\ast (n+1)}$ has any first order 
poles, and thus by collecting all terms with $j+k=m$ we obtain the recursive equation for $c_{m,n+1}$.
The calculation for $\tilde s_{jn}$, i.e.\ for $\vert \chi\vert >\varphi$, is very similar. The only difference is that the residue  of the pole at $q_{j+k}$ in (\ref{integral}) is $2\pi \ii$ instead of $-2\pi \ii$, a difference which is due to two jumps of $2\pi \ii$ of the logarithmic terms at the cut along the negative real axis. This gives the additional 
minus sign in the recursion for $\tilde c_{jn}$.

We turn to $L_{jkn}$. 
As mentioned above, there is no singularity at $q_{j+k}$. There are two
logarithmic singularities at $q_j$ and $q_k$. Other than that, $L_{jkn}$ is bounded. Set $\delta:=\vert \varphi-\vert \chi\vert \vert$ and $h(\delta):=\vert\ln(\sin\delta)\vert$.
Then $|L_{jkn}(z)| \leq c_{j1} c_{kn}\, K(1+ h( \delta))$ for some constant $K$. It is immediate from (\ref{def cjn})  that $c_{jn} = 0$ for $j < n$. Thus there exists some constant $B$, independent of $n$, such that
\[
\sum_{j,k=1}^N |L_{jkn}(z)| \leq B\,(1+h(\delta)).
\]
Now we tackle the second line of \eqref{g n+1}. By the induction hypothesis, 
$|L_n(z)| \leq A^n (1+h(\delta)^s)$. Thus
\[
\begin{split}
I_{jn}(z) & := \left| \int_0^z s_{j1}(\zeta) L_n(z-\zeta) \, \dd \zeta \right| \leq \\ 
& \leq A^n h(\delta)^s \int_0^s | s_{j1}(y\e{\ii\chi}) | \,h(\delta)^{-y} \, \dd y+A^n\int_0^s | s_{j1}(y\e{\ii\chi}) | 
\, \dd y.
\end{split}
\]
For $s \le\pi$, the integrands on the right hand side above are bounded, and hence $I_{jn}(z) \leq$
$ DA^n(1+h(\delta)^s)$ for some constant $D$. For $s > \pi$, we decompose the domain of integration into $y\le \pi$ and $\pi < y \le s$. On the first interval,
the integrands are bounded with the same result as above. On the second interval, we replace $h(\delta)^{-y}$
by $h( \delta)^{-\pi}$. The integral over $s_{1j}$ alone is clearly bounded by $E(1+ h(\delta))$ for some constant $E$,
and we estimate crudely 
\[
\sum_{j=1}^N
I_{jn}(z) \leq 3N A^n(D+E)(1+ h(\delta)^s).
\]
Finally, since $L_1$ is bounded, we have 
\[
\sum_{j=1}^N \left| \int_{0}^z s_{jn}(\zeta) L_1(z-\zeta) \, \dd \zeta \right|  \leq F (1+h(\delta)),
\]
where the constant $F$  does not depend on $n$. Altogether, we find 
$|L_{n+1}(z)| \leq 3NA^n(D+E+B+F) (1+h(\delta)^s)$.
Clearly, setting $A=3N(1+D+E+B+F)$, this is bounded by $A^{n+1}(1+ h(\delta)^s)$.  
\end{proof}

In order to extend this result to the right half-plane one has to take account of the poles $\overline{q}_j$ of $g$, too. This amounts in replacing $s_{jn}(z)$ by $s_{jn}(z)+\overline{s_{jn}(\overline{z})}$. Let $M_n$ denote the remainder in place of $L_n$. It satisfies the same kind of estimate with some new constant $A$. Using \eqref{dens} and the fact that $c_{jn} = 0$ for $j < n$, we now have 
\[
\varrho(z) = \sum_{j=1}^N \left(\sum_{n=1}^j \frac{c_{jn}}{n!}\right)\left(\frac{-1}{2 \pi \ii (z - q_j)}+\frac{1}{2\pi\ii(z-\overline{q}_j)}\right) + M(z)
\]
with $M:=\sum_{n=1}^\infty\frac{1}{n!} M_n$ and $\vert M(z)\vert\le \e{A}(1+h(\delta)^s)$ for $-\varphi<\chi<\varphi$.
 The same formula with $\tilde c_{jn}$ and 
$\tilde M$ holds for $\vert\chi\vert >\varphi$. Obviously $h(\delta)^s$ can be replaced by $\vert \ln\delta\vert^{2\pi N}$.
\begin{proposition}
Define $c_{jn}$ and $\tilde c_{jn}$ as in \eqref{def cjn}. Then, for all $j\in\bbN$, 
\[
\sum_{n=1}^j \frac{c_{jn}}{n!} = \binom{j+2}{2}, \qquad \text{and} \qquad \sum_{n=1}^j \frac{\tilde c_{jn}}{n!} =(-1)^{j+1} \binom{3}{j} 
\quad \text{ if } j \leq 3, \text{and } 0 \text{ otherwise.}
\]  
\end{proposition}

\begin{proof}
We introduce the generating functions $F_n(x) = \sum_{j=1}^\infty c_{jn} x^j$. Then, 
$F_1(x) = \sum_{j=1}^\infty \frac{3}{j} x^j = - 3 \ln(1-x)$, and $F_n = F_1^{\,n}$. The last statement 
follows from 
\[
F_1(x) F_n(x) = \sum_{j=1}^\infty \left( \sum_{k=1}^{j-1} c_{k1} c_{j-k,n} \right) x^j = 
\sum_{j=1}^\infty c_{j,n+1} x^{j} = F_{n+1}(x).
\]
Now recall  that $c_{jn} = 0$ for $j < n$. Then
\[
\begin{split}
\sum_{n=1}^\infty \frac{c_{jn}}{n!} & = \frac{1}{j!} \partial_x^j \sum_{n=1}^\infty \frac{F_1(x)^n}{n!} \Big\vert_{x=0}
= \frac{1}{j!} \partial_x^j \left( \e{F_1(x)}-1 \right) \Big \vert_{x=0} = \\
& = \frac{1}{j!} \partial_x^j \e{- 3 \ln(1-x)} \Big\vert_{x=0} = \frac{1}{j!} \partial_x^j \frac{1}{(1-x)^3} 
\Big\vert_{x=0} = \binom{j+2}{2}.
\end{split}
\]
For the generating function $\tilde F_n$ of $(\tilde c_{jn})$, a similar calculation leads to 
$\tilde F_n = (-1)^{n+1} \tilde F_1^{\,n}$ with $\tilde F_1=F_1$. Thus, as above, 
\[
\sum_{n=1}^\infty \frac{\tilde c_{jn}}{n!} = - \frac{1}{j!}\partial_x^j  \e{3 \ln (1-x)} \Big\vert_{x=0} = 
- \frac{1}{j!} \partial_x^j (1-x)^3 \Big\vert_{x=0}.
\]
This equals $3$ for $j=1$, $-3$ for $j=2$, $1$ for $j=3$, and zero otherwise, as was claimed. 
\end{proof}

The stated analyticity properties of $h_N(z)$ and $\tilde h_N(z)$ from Theorem \ref{ems structure 2} 
are  true by the fact that both $\phi$ 
(cf. Theorem \ref{ems structure 1}) and $\ell_j$ possess them.  This concludes the proof of Theorem \ref{ems structure 2}. 
%


\section{Appendix}

Here we show that Definition \ref{reso} of a resonance  is equivalent to that given e.g.\ in 
\cite[XII.6]{RS}. We start with a Lemma that connects the analyticity of a measure's density with 
properties of its Stieltjes transform.

Let $\mu$ be a finite Borel measure on $[0,\infty[$ and $f:=\tilde{\mu}$  its Stieltjes transform. Note that $f$ is holomorphic on $\bbC\setminus[0,\infty[$ and satisfies $\overline{f(z)}=f(\overline{z})$.  We recall the classical inversion formula valid for $t\ge 0$:
\be\label{stieltjes}
\mu([0,t])=\operatorname{lim}_{\delta\to 0+}\operatorname{lim}_{\epsilon\to 0+}\frac{1}{\pi}\int_0^{t+\delta}\operatorname{Im}\,f(s+\ii \epsilon)\,\dd s.
\ee
\begin{lemma}\label{analcont}
Let $t_0>0$ and $U\subset\bbC$ an open disc centered at $t_0$ with radius smaller than $t_0$. Then the following two 
statements are equivalent: 
\begin{itemize}
\item[(i)] there is an holomorphic function $F$ on $U$ which equals $f$ on $\{z\in U:\operatorname{Im} z>0\}$.
\item[(ii)] $\mu$ on $U\cap \bbR$ is absolutely continuous with respect to Lebesgue measure, 
and its density  is the restriction on $U\cap \bbR$ of a holomorphic function $\phi$ on $U$.
\end{itemize}
If (i) and (ii) hold, then obviously $f(t):=\operatorname{lim}_{\epsilon\to 0+}f(t+\ii \epsilon)$ converges uniformly on compact subsets of $U\cap \bbR$, and $F = f + 2\pi \ii \phi $ holds on $\{z\in U: \operatorname{Im} \,z<0\}$.
\end{lemma}
\begin{proof}
Assume first that $F$ exists. Then (\ref{stieltjes}) yields immediately for $t_1,\,t_2$ in $U\cap\bbR,\,t_1<t_2$: $\mu(]t_1,t_2])=\frac{1}{\pi}\int_{t_1}^{t_2}\operatorname{Im} F(s)\dd s$. This implies that $\mu$ on $U\cap \bbR$  is absolutely continuous with respect to Lebesgue measure and that the density  is given by $t\to \frac{1}{\pi}\operatorname{Im}F(t)$. Then $\phi(z):=\frac{1}{2\pi \ii}(F(z)-\overline{F(\overline{z})})$ is its analytic continuation  on $U$. For $z\in U$ with $\operatorname{Im}\,z<0$ one has $\overline{F(\overline{z})}=\overline{f(\overline{z})}=f(z)$, whence $F(z)=f(z)+2\pi \ii \phi(z)$.

Now assume the existence of $\phi$. For $z\in U$ set $F(z):=f(z)$ if $\operatorname{Im} z>0$ and $F(z):=f(z)+2\pi \ii 
\phi(z)$ if $\operatorname{Im} z<0$. Then $F$ is holomorphic on $U\setminus \bbR$. 
Fix $t\in U\cap \bbR$. Let   $\delta>0$ such that $\{t-\delta,\,t+\delta\}\subset U$. We define three paths in $\bbC$. 
First $\gamma_1(s):=s$ for $s\in [0,\infty[$. Then $\gamma_2$ differs from $\gamma_1$ only in that it joins the point $t-
\delta$ to $t+\delta$ not by the straight line but by the semi-circle through $t+i\delta$.  Finally the closed path $
\gamma_3$ joins  $t-\delta$ to $t+\delta$ forwards by the straight line and backwards by the semi-circle through $t+\ii
\delta$. Then for $0<\epsilon<\delta$ one has
\[
\int_{\gamma_1}\frac{\dd\mu(z)}{z-(t+\ii\epsilon)}-\int_{\gamma_2}\frac{\dd\mu(z)}{z-(t+\ii\epsilon)} =\int_{\gamma_3}\frac
{\phi(z)\dd z}{z-(t+\ii\epsilon)}=2\pi \ii \,\phi (t+\ii\epsilon)
\]
by the residue theorem. Therefore
\[
\operatorname{lim}_{\epsilon\to 0+}f(t+\ii\epsilon)=\int_{\gamma_2}\frac{\dd\mu(z)}{z-t} +2\pi \ii\,\phi(t)
\]
exists.  Similarly $\operatorname{lim}_{\epsilon\to 0+}f(t-\ii\epsilon)$ is shown to exist. 

Set $F(t):=\operatorname{lim}_{\epsilon\to 0+}f(t+\ii\epsilon)$. Thus $F$ is defined on the whole of $U$. It remains to 
show that $F$ stays holomorphic. By the following lemma
$F(t+\ii\epsilon)-F(t-\ii\epsilon)=f(t+\ii\epsilon)-f(t-\ii\epsilon)-2\pi \ii \, \phi(t)+2\pi \ii\left(\phi(t)-\phi(t-\ii
\epsilon)\right)\to 0$ as $\epsilon\to 0+$ uniformly on compact subsets of $U\cap\bbR$. From this the premises on $F$ of Morera's theorem easily follow, whence the result.
\end{proof}

Now the equivalence of Definition \ref{reso} with the traditional definition of resonances follows by applying 
Lemma \ref{analcont}  to $\mu=\mu^{(u)}$ and $\mu =\mu_0^{(u)}$: choosing $U$ such that 
the density $\phi$ of $\mu^{(u)}$ or $\mu_0^{(u)}$ is analytic on $U$, we find that the continuation $F$ on $U$ of the 
resolvent to the second Riemann sheet  is given by $f + 2 \pi \ii \phi$.

We close with a lemma showing that a strong version of \eqref{stieltjes} holds if the  density of $\mu$ is continuously differentiable.

\begin{lemma}\label{uniconv}
Let $\mu$ be absolutely continuous with respect to Lebesgue measure on $]A,B[$ for $0\le A<B$ and let the density $\phi$ be continuously differentiable.  Then $\phi(t)=\operatorname{lim}_{\epsilon\to 0+}$ $\frac{1}{\pi}\operatorname{Im}f(t+i\epsilon)$ holds uniformly on every compact subset of  $]A,B[$.
\end{lemma}
\noindent{\it Proof.} Put $\chi(s;t,\epsilon):=\frac{\epsilon/\pi}{(s-t)^2+\epsilon^2}$. 
Then $\frac{1}{\pi}\operatorname{Im}f(t+\ii\epsilon)=\int_{[0,\infty[} \chi(s;t,\epsilon)\dd \mu(s)$.
 Let $A<A_1<B_1<B$ and $t\in[A_1,B_1]$. 
 
i) First $\int_{[B,\infty[} \chi(s;t,\epsilon)\dd \mu(s) \le \frac{\epsilon/\pi}{(B-B_1)^2}\int_{[0,\infty[}\dd  \mu(s)\to 0$ uniformly  as $\epsilon\to 0+$. Similarly this holds for 
$\int_{[0,A] }\chi(s;t,\epsilon)\dd \mu(s)$.

ii) Next 
\[ 
\begin{split}
1\ge \int_A^B \chi(s;t,\epsilon)\,\dd s & = 
\frac{1}{\pi}(\operatorname{arctan}\frac{B-t}{\epsilon}-\operatorname{arctan}\frac{t-A}{\epsilon})\ge \\
& \geq \frac{1}{\pi}(\operatorname{arctan}\frac{B-B_1}{\epsilon}-\operatorname{arctan}\frac{A_1-A}{\epsilon}) \to 1
\end{split}
\]
as $\epsilon\to 0+$. This implies that $\int_A^B \chi(s;t,\epsilon)\,\dd s$ tends uniformly to $1$ as $\epsilon\to 0+$.

iii) Because of i) it remains to show that $\vert\int_A^B\phi(s) \chi(s;t,\epsilon) \dd s-\phi(t)\vert\le\vert\int_A^B(\phi(s)-\phi(t))\chi(s;t,\epsilon)\, \dd s\vert+\phi(t)\vert\int_A^B\chi(s;t,\epsilon)\dd s-1\vert$ tends uniformly to $0$ as $\epsilon\to 0+$. This holds true for the second summand because of ii). As to the first summand  the mean value theorem yields $\vert\phi(s)-\phi(t)\vert\le c\vert s-t\vert$ with $c:=\operatorname{sup}\{\phi'(\tau):\tau\in [A_1,B_1]\}$. This finishes the proof since $\int_A^B\vert s-t\vert \chi(s;t,\epsilon)\,\dd s=\frac{\epsilon}{\pi}\int_{A-t}^{B-t}\vert r\vert(r^2+\epsilon^2)^{-1}\dd r\le \frac{2\epsilon}{\pi}\int_0^Br(r^2+\epsilon^2)^{-1}\dd r =$ $\frac{\epsilon}{\pi} \operatorname{ln}\frac{B^2+\epsilon^2}{\epsilon^2}\to 0$ as $\epsilon \to 0+$.\hfill{$\Box$}

\end{document}